\documentclass{scrartcl}

\KOMAoptions{paper=letter} 

\usepackage[todo,firstnames]{group} 
\usepackage{wasysym}
\usepackage{bm}
\usepackage{thm-restate}
\usepackage{blkarray}
\usepackage{etoolbox}
\usetikzlibrary{matrix,decorations.pathreplacing,calc,positioning,fit}

\setkomafont{caption}{\footnotesize}
\setkomafont{captionlabel}{\normalsize} 

\newcommand{\nash}[1][]{\ifthenelse{\equal{#1}{}}{\ensuremath{\mathit{NASH}}\xspace}{\ensuremath{\mathit{NASH}(#1)}\xspace}}
\newcommand{\ceq}[1][]{\ifthenelse{\equal{#1}{}}{\ensuremath{\mathit{CEQ}}\xspace}{\ensuremath{\mathit{CEQ}(#1)}\xspace}}
\newcommand{\eps}{\varepsilon}
\newcommand{\mcal}{\mathcal}

\newcommand{\setm}{\setminus}
\newcommand{\timesdots}{\times\dots\times}

\makeatletter

\DeclareMathOperator{\hada}{\ast}

\newtheorem{case}{Case}

\title{A Robust Characterization of\\ Nash Equilibrium}

\author{Florian Brandl\\\normalsize Department of Economics\\[-1ex]\normalsize University of Bonn%
\and 
Felix Brandt%
\\\normalsize Department of Computer Science\\[-1ex]
\normalsize Technical University of Munich%
}

\begin{document}

\maketitle

\begin{abstract}

We characterize Nash equilibrium by postulating coherent behavior across varying games. Nash equilibrium is the only solution concept that satisfies the following axioms:
\emph{(i)} strictly dominant actions are played with positive probability,
\emph{(ii)} if a strategy profile is played in two games, it is also played in every convex combination of these games, and
\emph{(iii)} players can shift probability arbitrarily between two indistinguishable actions, and deleting one of these actions has no effect. 
Our theorem implies that every equilibrium refinement violates at least one of these axioms. 
Moreover, every solution concept that approximately satisfies these axioms returns approximate Nash equilibria, even in natural subclasses of games, such as two-player zero-sum games, potential games, and graphical games.
\end{abstract}

 \textbf{Keywords:} Game theory, axiomatic characterization, Nash equilibrium

\section{Introduction}

More than 70 years after the publication of \citeauthor{Nash51a}'s (\citeyear{Nash51a}) original work, the concept of Nash equilibrium has been engraved in economic reasoning so deeply that it is rarely questioned. 
But what makes Nash equilibrium stand out from the plethora of solution concepts that have been proposed?
The game theory literature has come up with various answers to this question based on different approaches.

In this paper, we take an axiomatic approach: we consider solution concepts for games in normal form with a fixed number of players and formulate conditions for solution concepts that capture coherent behavior across different games.
The solution concepts we consider map every (finite) multi-player game to a non-empty set of (mixed) strategy profiles.
Our first of three axioms requires that the labels of actions are irrelevant---only the payoffs matter.
Call two actions of a player clones if, irrespective of the other players' actions, they give the same payoff to all players.
In other words, clones are outcome-equivalent and only discernible by their labels.
\begin{quote}
	\textbf{Consequentialism.}
	A player can shift probability arbitrarily between clones, and deleting a clone neither changes the probabilities assigned to the player's other actions nor the strategies of the other players.
\end{quote}
If a solution concept satisfying consequentialism returns a strategy profile and we modify the game by cloning an action, then the solution concept has to return all strategy profiles in which the probabilities on the player's uncloned actions and the other players' strategies are unchanged.

The second axiom is motivated by situations where the players are uncertain which game will be played.   
\begin{quote}
	\textbf{Consistency.}
	Every strategy profile that is played in two given games with the same sets of actions for each player is also played when a coin toss decides which of the two games is played and the players choose their strategies before the coin toss.
\end{quote}
Instead of modeling the randomization explicitly, we assume that a coin toss between two games is equivalent to the convex combination of these games.

Third, we stipulate a very weak notion of rationality.
\begin{quote}
	\textbf{Rationality.}
	An action that dominates every other action of the same player in pure strategies is played with non-zero probability.
\end{quote}
This notion of rationality is, for example, weaker than the condition demanding that actions that are dominated in pure strategies are never played.

Our main result characterizes Nash equilibrium as the unique solution concept that satisfies consequentialism, consistency, and rationality.
In particular, players' behavior has to be consistent with expected utility maximization, which is not apparent from the axioms.
Moreover, every refinement of Nash equilibrium violates at least one of the axioms.

Our second result shows that this characterization is robust: every solution concept that approximately satisfies the three axioms is approximately Nash equilibrium.
This type of stability result is common in many areas of mathematics.\footnote{A classic example is isoperimetric stability.
In Euclidean space, a ball is characterized as the unique volume-maximizing shape among all well-behaved shapes with the same surface area.
Isoperimetric stability strengthens this assertion by showing that any shape that is close to volume maximizing has to resemble a sphere.}
To make it precise, we formulate quantitatively relaxed versions of the axioms.
First, $\delta$-consequentialism demands that a player can shift probability arbitrarily between clones and deleting a clone does not change the probability on any player's action (except the cloned action) by more than $\delta$.
Second, $\delta$-consistency requires that if a strategy profile is played in two games, some strategy profile that differs by no more than $\delta$ is played in any convex combination of the two games.
Third, $\delta$-rationality implies that actions that are dominated in pure strategies by at least a margin of $\delta$ are played with probability at most, say, one half.
We show that for every positive $\eps$, there exists a positive $\delta$ such that every solution concept that satisfies the $\delta$-versions of consequentialism, consistency, and rationality and behaves well with respect to renaming a player's actions is a refinement of $\eps$-Nash equilibrium. 
This result implies one direction of our main theorem when slightly strenthening rationality.\footnote{More precisely, the implication holds when strenthening rationality to require that dominated actions are played with probability at most $\nicefrac12$.}
Moreover, it holds for various natural subclasses of games such as two-player zero-sum games, potential games, and graphical games. 

The remainder of the paper is structured as follows. Existing axiomatic work on Nash equilibrium is discussed in \Cref{sec:relatedwork}. We define the model and introduce the required notation in \Cref{sec:model}. In \Cref{sec:nashcharacterization}, we formally define the three axioms and prove that every solution concept that satisfies these axioms has to return Nash equilibria. The converse statement, i.e., that any such solution concept has to return \emph{all} Nash equilibria, is shown in the appendix. \Cref{sec:robustness} states the robust version of the characterization. \Cref{sec:discussion} concludes the paper by discussing consequences and variations of our results.  
 
\section{Related Work}\label{sec:relatedwork}

Which assumptions can be used to justify Nash equilibrium has been primarily studied in epistemic game theory.
In this stream of research, the knowledge of individual players is modeled using Bayesian belief hierarchies, which consist of a game and a set of types for each player, with each type including the action played by this type and a probability distribution over types of the other players, called the belief of this type \citep{Hars67a}. Rather than assuming that players actively randomize, the beliefs about the types of the other players are randomized. Players are rational if they maximize expected payoff given their types and beliefs.
\citet{AuBr95a} have shown that for two-player games, the beliefs of every pair of types whose beliefs are mutually known and whose rationality is mutually known constitute a Nash equilibrium.
This result extends to games with more than two players if the beliefs are commonly known and admit a common prior. Common knowledge assumptions in game theory have been criticized for not adequately modeling reality \citep[see, e.g.,][]{Gint09a}.
\citet{Bare09a}, \citet{Hell13a}, and \citet{BaTs14a} showed that the results of \citet{AuBr95a} still hold under somewhat weaker common knowledge assumptions.

Building on earlier work by \citet{PeTi96a}, \citet{NPRV96a} have characterized Nash equilibrium via one-player rationality (only utility-maximizing strategies are returned in one-player games) and a consistency condition that is orthogonal to ours because it varies the set of \emph{players}. Their condition requires that every strategy profile $s$ returned for an $n$-player game is also returned for the $(n-k)$-player game that results when $k$ players invariably play their strategies in $s$. The two axioms immediately imply that only subsets of Nash equilibria can be returned. Their results have no implications for games with a fixed number of players. 
Other axiomatic work on Nash equilibrium includes a characterization of pure Nash equilibrium \citep{Voor19a} and a characterization of Nash equilibrium for games with quasiconcave utility functions \citep{Salo92a}.

The work most closely related to ours is due to \citet{BrBr17c}, who have characterized \emph{maximin strategies} in two-player zero-sum games by consequentialism, consistency, and rationality. The differences between their results and ours are as follows.
Solution concepts as considered by \citeauthor{BrBr17c} return a set of strategies for one player rather than a set of strategy \emph{profiles}. 
Noting that Nash equilibria in zero-sum games consist of pairs of maximin strategies, their result translates as follows in the terminology of the present paper:
\emph{in zero-sum games}, every solution concept that satisfies consequentialism, consistency, and rationality returns an (exchangeable) \emph{subset} of Nash equilibria.\footnote{A set of strategy profiles is exchangeable if it is a Cartesian product of a set of strategies for each player.}
Our main theorem, \Cref{thm:nash}, is stronger since it \emph{(i)} holds for any number of players, \emph{(ii)} shows that \emph{all} Nash equilibria have to be returned (and thus rules out equilibrium refinements), and \emph{(iii)} is not restricted to games with rational-valued payoffs and rational-valued strategies (which are assumptions required for the proof of \citeauthor{BrBr17c}).
Moreover, we show that the containment in the set of Nash equilibria \emph{(iv)} also holds for restricted classes of games (cf.\ \Cref{sec:discussion}).

\section{The Model}\label{sec:model}

Let $U$ be an infinite universal set of actions and denote by $\mathcal F(U)$ the set of finite and nonempty subsets of $U$.
A permutation of $U$ is a bijection from $U$ to itself that fixes all but finitely many elements, and for $A\in\mathcal F(U)$, $\Sigma_{A}$ is the set of permutations of $U$ that fix each element of $U\setm A$.
We denote by $\mathbb R_+$ and $\mathbb R_{++}$ the set of non-negative and positive real numbers, respectively, and we use similar notation for $\mathbb Q$ and $\mathbb Z$.
If $p\in\mathbb R^U$, we write $\|p\| = \sum_{a\in U} |p(a)|$ for the $\ell_1$-norm of $p$ whenever the sum on the right-hand-side is finite, and we write $\supp(p) = \{a\in U\colon p(a) \neq 0\}$ for the support of $p$.
Moreover, let $$\Delta A = \{p\in \mathbb R_+^U\colon \supp(p)\subseteq A\text{ and }\sum_{a\in A} p(a) = 1\}$$ be the set of probability distributions on $U$ that are supported on $A$.
We call $\Delta A$ the set of strategies for action set $A$.
 The ball of radius $\delta > 0$ around a set $S\subseteq \Delta A$ is $B_\delta(S) = \{p\in \Delta A\colon \inf\{\|p - q\|\colon q \in S\} <\delta\}$, the set of strategies supported on $A$ and less than $\delta$ away from some strategy in $S$.\footnote{Note that $B_\delta(S)$ depends on $A$. In our usage, $A$ will be clear from the context.}

Let $N = \{1,\dots,n\}$ be the set of players.
For action sets $A_1,\dots,A_n\in\mcal F(U)$, we write $A = A_1\times\dots\times A_n$ for the corresponding set of action profiles.
A game on $A$ is a function $G\colon A\rightarrow \mathbb R^n$.
For $i\in N$ and $a\in A$, $G_i(a)$ is the payoff of player $i$ for the action profile $a$.
 We say that $G$ is normalized if for every player $i$, either $G_i$ has minimum 0 and maximum 1 or is constant at 1.
We call $\Delta A_1\times\dots\times\Delta A_n$ the set of (strategy) profiles on $A$.
 The ball of radius $\delta > 0$ around a set $S$ of strategy profiles, $B_\delta(S)$, is defined via the norm $\|p - q\| = \max_{i\in N} \|p_i - q_i\|$ on strategy profiles.
The players' payoffs for a strategy profile are the corresponding expected payoffs. 
Thus, a strategy profile $p$ is a Nash equilibrium of $G$ if
\begin{align*}
	G_i(p_i,p_{-i}) \ge G_i(q_i,p_{-i}) \text{ for all $q_i\in\Delta A_i$ and $i\in N$.}
\end{align*}

For two games $G$ and $G'$ on $A = A_1\timesdots A_n$ and $A' = A_1'\timesdots A_n'$, we say that $G$ is a blow-up of $G'$ or that $G'$ is a blow-down of $G$ if $G$ can be obtained from $G'$ by replacing actions with multiple payoff-equivalent actions and renaming actions.
That is, there are surjective functions $\phi_i\colon A_i\rightarrow A_i'$, $i\in N$, such that with $\phi = (\phi_1,\dots,\phi_n)$, $G = G'\circ\phi$.
Actions in $\phi_i^{-1}(a_i')$ for $a_i'\in A_i'$ are called clones of $a_i'$.
So $G$ is obtained from $G'$ by replacing each action $a_i'$ by $|\phi_i^{-1}(a_i')|$ clones of it.\footnote{\citet{KoMe86a} have considered a more permissive notion of ``blowing down'' in the context of Nash equilibrium refinements for extensive-form games. Their notion of a reduced form of a normal-form game allows deleting any action that is a convex combination of other actions.} 
A strategy $p_i\in \Delta A_i$ induces a strategy on $A_i'$ via the pushforward along $\phi_i$: $(\phi_i)_*(p_i) = p_i\circ\phi^{-1}$.
Then, a strategy profile $p$ on $A$ induces the strategy profile $\phi_*(p) = ((\phi_1)_*(p_1),\dots,(\phi_n)_*(p_n))$ on $A'$.

A solution concept $f$ maps every game $G$ to a set of strategy profiles $f(G)$ on the actions of $G$.
If $f(G)\neq \emptyset$ for all $G$, $f$ is a \emph{total} solution concept.
An example of a solution concept is \nash, which returns all strategy profiles that constitute Nash equilibria. \citet{Nash51a} has shown that every game admits at least one Nash equilibrium, and so $\nash$ is total.
A non-example of a solution concept is correlated equilibrium since correlated strategy profiles (i.e., distributions over action profiles) are not strategy profiles according to our definition.

\section{Characterization of Nash Equilibrium}\label{sec:nashcharacterization}

This section defines our axioms and states the characterization of Nash equilibrium along with the more illuminating part of its proof.
The remainder of the proof and all other proofs are given in the Appendix.

Consequentialism requires that if $G$ is a blow-up of $G'$, a strategy profile is returned in $G$ if and only if its pushforward is returned in $G'$.
Equivalently, it asserts that \emph{(i)} cloning an action does not change the probabilities of other actions and the strategies of the other players, and \emph{(ii)} the probability on the cloned action can be distributed arbitrarily among its clones.

\begin{definition}[Consequentialism]\label{def:consequentialism}
	A solution concept $f$ satisfies consequentialism if for all games $G$ and $G'$ such that $G$ is a blow-up of $G'$ with surjection $\phi = (\phi_1,\dots,\phi_n)$,
	\begin{align*}
		f(G) = \phi_*^{-1}(f(G')).
	\end{align*}
\end{definition}
Consequentialism is a common desideratum in decision theory.
It corresponds to the conjunction of \citeauthor{Cher54a}'s (\citeyear{Cher54a}) \emph{Postulate 6} (cloning of a player's actions) and \emph{Postulate 9} (cloning of Nature's states, i.e., of opponent's actions). 
The latter also appears as \emph{column duplication} \citep{Miln54a} and \emph{deletion of repetitious states} \citep{ArHu72a,Mask79a}.
In the context of social choice theory, a related condition called \emph{independence of clones} was introduced by \citet{Tide87a} \citep[see also][]{ZaTi89a,Bran13a}.

Suppose $G = G'$ and $\phi_i$ permutes the actions of each player $i$.
Then consequentialism reduces to equivariance, that is, relabeling the actions of a player results in the same relabeling of her strategies.
Formally, a solution concept $f$ satisfies \emph{equivariance} if for all games $G$ on $A$ and all $\pi = (\pi_1,\dots,\pi_n)$ where $\pi_i$ is a permutation of $A_i$,\footnote{For a strategy $p_i\in\Delta(A_i)$, $p_i\circ\pi_i$ is the strategy with $(p_i\circ\pi_i)(a_i) = p_i(\pi(a_i))$. For a strategy profile $p = (p_1,\dots,p_n)$, $p\circ\pi = (p_1\circ\pi_1,\dots,p_n\circ\pi_n)$, and this operation extends to sets of strategy profiles pointwise.}  
	\begin{align*}
		f(G\circ\pi) = f(G) \circ \pi.
	\end{align*}
We will frequently apply equivariance to strategy profiles where each player's strategy is the uniform distribution on some subset of her actions, and the permutations map each action to an action with the same probability, thus giving a new game for which the same strategy profile is returned.

Consistency requires that if a strategy profile is returned in two games with the same action sets, it is also returned in any convex combination of these games. An inductive argument shows that this is equivalent to the extension of the axiom to convex combinations of any finite number of games. We will frequently use this fact in our proofs.

\begin{definition}[Consistency]\label{def:consistency}
	A solution concept $f$ satisfies consistency if for any two games $G,G'$ on $A$ and any $\lambda\in[0,1]$,
	\begin{align*}
		f(G)\cap f(G')\subseteq f(\lambda G + (1-\lambda) G').
	\end{align*}
\end{definition}

We are not aware of game-theoretic work using this consistency axiom other than that by \citet{Bran13a}.
\citeauthor{Cher54a} considers combinations of decision-theoretic situations obtained by taking unions of action sets.
His \emph{Postulate 9} states that any action that is chosen in two situations should also be chosen in such a combination.
In our context, this translates to a consistency condition on the support of strategies and varying sets of actions.
Closer analogs of consistency, involving convex combinations of distributions over states (i.e., strategies of Nature), have been considered as decision-theoretic axioms \citep[see, e.g.,][]{Cher54a,Miln54a,GiSc03a}.
\citeauthor{Shap53c}'s (\citeyear{Shap53c}) characterization of the Shapley value involves an additivity axiom (which he calls \emph{law of aggregation}) that is similar in spirit to consistency.
Lastly, analogs of consistency feature prominently in several axiomatic characterizations in social choice theory, where it relates the choices for different sets of voters to each other \citep[see, e.g.,][]{Smit73a,Youn75a,YoLe78a,Myer95b,Bran13a,LaSk21a}.

For a game $G$ on $A$ and two actions $a_i,a_i'\in A_i$, we say that $a_i$ \emph{dominates} $a_i'$ if $G_i(a_i,a_{-i}) > G_i(a_i',a_{-i})$ for all $a_{-i}\in A_{-i}$; $a_i$ is \emph{dominant} if it dominates every other action in $A_i$.
Clearly, dominant actions are unique whenever they exist.
Rationality requires that a dominant action has to be played with non-zero probability.

\begin{definition}[Rationality]\label{def:rationality}
	A solution concept $f$ satisfies rationality if for all games $G$, all $i\in N$, and all dominant $a_i\in A_i$,
	\begin{align*}
		(p_1,\dots,p_n) \in f(G)\text{ implies } p_i(a_i) > 0.
	\end{align*}
\end{definition}
Note that rationality is not concerned with \emph{mixed} strategies and thus does not rely on \emph{expected} payoffs.
Moreover, it does not need any assumptions about other players.
The strengthening of rationality requiring that dominated action receive probability~0 is equivalent to \citeauthor{Miln54a}'s (\citeyear{Miln54a}) \emph{strong domination}, \citeauthor{Mask79a}'s (\citeyear{Mask79a}) \emph{Property (5)}, and weaker than \citeauthor{Cher54a}'s (\citeyear{Cher54a}) \emph{Postulate 2}.

It turns out that Nash equilibrium is the only total solution concept that satisfies the three axioms defined above.

\begin{restatable}{theorem}{nashtheorem}\label{thm:nash}
	Let $f$ be a total solution concept that satisfies consequentialism, consistency, and rationality.
	Then, $f = \nash$.
\end{restatable}

The proof of \Cref{thm:nash} uses a lemma that illustrates how one can manipulate games using the above axioms.
It shows that solution concepts satisfying consequentialism and consistency behave as one would hope under the analog of row and column operations familiar from linear algebra.
More precisely, it shows that when adding a linear combination of some actions (with positive rational-valued coefficients) to another action, then a solution concept satisfying consequentialism and consistency shifts probability from the former actions to the latter in proportion to the coefficients (see \Cref{fig:lem:linearalgebraillustration}).
A linear combination of actions here means a linear combination of the corresponding payoffs for all players.
A similar conclusion applies to adding new actions that are linear combinations of existing ones.

\begin{figure}[tbp]
 	\centering
	\begin{tikzpicture}[
  baseline,
  label distance=0pt 
]

\matrix [anchor = north west, matrix of math nodes, left delimiter=(,right delimiter=), row sep=.1cm, column sep=.1cm] (g) {
      1,0 & 1,1 & 1,0\\  
		0,1 & 2,0 & 0,1\\ 
		1,1 & 1,0 & 0,0\\ 
		1,0 & 0,0 & 1,1\\};

\node[above = 1.0cm of g.north] (glabel) {$G$};

\node[
  fit=(g-1-1),
  inner xsep=0,
  label=above:$\nicefrac13$
 ] {};
\node[
  fit=(g-1-2),
  inner xsep=0,
  label=above:$\nicefrac13$
 ] {};
\node[
  fit=(g-1-3),
  inner xsep=0,
  label=above:$\nicefrac13$
 ] {};
 
\node[
  fit=(g-1-1),
  inner xsep=13pt,
  label=left:$0$
 ] {};
\node[
  fit=(g-2-1),
  inner xsep=13pt,
  label=left:$\nicefrac13$
 ] {};
\node[
  fit=(g-3-1),
  inner xsep=13pt,
  label=left:$\nicefrac13$
 ] {};
\node[
  fit=(g-4-1),
  inner xsep=13pt,
  label=left:$\nicefrac13$
 ] {};

\matrix[right = 3.4cm of g.north east, anchor = north west, matrix of math nodes,left delimiter=(,right delimiter=), row sep=.1cm, column sep=.1cm] (ghat) {
      1,0 & 1,1 & 1,0\\  
		0,1 & 2,0 & 0,1\\ 
		1,1 & 1,0 & 0,0\\ 
		\frac34,\frac12 & \frac34,0 & \frac14,\frac34\\};

\node[above = 1.0cm of ghat.north] (ghatlabel) {$\hat G$};

\node[
  fit=(ghat-1-1),
  inner xsep=0,
  label=above:$\nicefrac13$
 ] {};
\node[
  fit=(ghat-1-2),
  inner xsep=0,
  label=above:$\nicefrac13$
 ] {};
\node[
  fit=(ghat-1-3),
  inner xsep=0,
  label=above:$\nicefrac13$
 ] {};
 
\node[
  fit=(ghat-1-1),
  inner xsep=13pt,
  label=left:$0$
 ] {};
\node[
  fit=(ghat-2-1),
  inner xsep=13pt,
  label=left:$\nicefrac16$
 ] {};
\node[
  fit=(ghat-3-1),
  inner xsep=13pt,
  label=left:$\nicefrac16$
 ] {};
\node[
  fit=(ghat-4-1),
  inner xsep=13pt,
  label=left:$\nicefrac23$
 ] {};

\end{tikzpicture} 
\caption{Example for an application of \Cref{lem:linearalgebra}.
Here, one half of the second and third action of the first player are added to the fourth action.
That is, $\hat a_1$ is the fourth action, $k_1 = (0,1,1,2)$, $\kappa_1 = \nicefrac16$, and $x_1 = (0,\nicefrac16,\nicefrac16,\nicefrac13)$; $\hat a_2$ is arbitrary, say, the first action of player 2, $k_2 = (1,0,0)$, $\kappa_2 = 1$, and $x_2 = (1,0,0)$.
}
\label{fig:lem:linearalgebraillustration}
\end{figure}

\begin{lemma}\label{lem:linearalgebra}
	Let $f$ be a solution concept satisfying consequentialism and consistency, $G$ be a game on $A$, and $p\in f(G)$.
	Let $\hat a = (\hat a_1,\dots,\hat a_n)\in U^N$, $k = (k_1,\dots,k_n)\in\mathbb Z_+^{U\times N}$, and $\kappa = (\kappa_1,\dots,\kappa_n)\in\mathbb R_{+}^N$ such that for all $i\in N$, $k_i \neq(0,0,\dots)\in\mathbb Z_+^U$, $k_i(\hat a_i) > 0$ if $\hat a_i\in A_i$, $\supp(k_i)\subseteq A_i$, $x_i:=\kappa_i k_i \le p_i$, and $x_i(\hat a_i) = p_i(\hat a_i)$.
	Then, there is a game $\hat G$ on $\hat A$ with $\hat A_i = A_i\cup\{\hat a_i\}$ so that the following holds.
	\begin{enumerate}
		\item $\hat p \in f(\hat G)$, where $\hat p_i = p_i - x_i + \|x_i\| e_{\hat a_i}$.\footnote{By $e_{\hat a_i}$, we denote the standard unit vector in $\mathbb R^U$ with a 1 in position $\hat a_i$.}
		\label{item:linalg1}
		\item For all $I\subseteq N$ and $a\in A$, $\hat G(\hat a_I,a_{-I}) = G\left(\left(\frac{k_i}{\|k_i\|}\right)_{i\in I}, a_{-I}\right)$.
		\label{item:linalg2}
	\end{enumerate}
\end{lemma}

Condition~\ref{item:linalg1} states that $\hat p_i$ is obtained from $p_i$ by shifting probability $x_i(a_i)$ from $a_i$ to $\hat a_i$ for all $a_i\neq \hat a_i$.
Since $x_i \le p_i$, $\hat p_i$ is a lottery, and since $x_i$ is a scalar multiple of an integer-valued vector, the ratios between the shifted probabilities are rational numbers, which is crucial for the proof technique.
Condition~\ref{item:linalg2} ensures that playing $\hat a_i$ in $\hat G$ is payoff-equivalent to playing $\frac{k_i}{\|k_i\|}$ in $G$.
We remark that for each $i\in N$, one of three cases occurs: 
\begin{enumerate}[leftmargin=*,label=\textit{(\roman*)}]
	\item $\hat a_i \in A_i$ and $p_i(\hat a_i) > 0$: then, $k_i(\hat a_i) > 0$ and $\kappa_i > 0$, and so $\hat a_i$ is replaced by a linear combination of actions in $A_i$, each with positive probability in $p_i$, and non-zero weight on $\hat a_i$, and probability $x_i(a_i)$ is shifted from $a_i$ to $\hat a_i$ for all $a_i\neq \hat a_i$.
	\item $\hat a_i\in A_i$ and $p_i(\hat a_i) = 0$: then, $k_i(\hat a_i) > 0$ and $\kappa_i = 0$, and $\hat a_i$ is replaced by a linear combination of actions in $A_i$ with non-zero weight on $\hat a_i$, and $\hat p_i = p_i$.
	\item $\hat a_i\in U\setminus A_i$ (and thus $p_i(\hat a_i) = 0$): then, $k_i(\hat a_i) = 0$ and $\kappa_i \ge 0$, and $\hat a_i$ is replaced by a linear combination of actions in $A_i$, and $\hat p_i = p_i$ if (and only if) $\kappa_i = 0$.\label{item:linearalgebra:case3}
\end{enumerate}
\Cref{fig:lem:linearalgebra:proof} illustrates the proof of \Cref{lem:linearalgebra} in Case~\ref{item:linearalgebra:case3}.

\begin{figure}[tbp]
 	\centering
	\begin{tikzpicture}[
  baseline,
  label distance=10pt 
]

\matrix [anchor = north west, matrix of math nodes, left delimiter=(,right delimiter=), row sep=.1cm, column sep=.1cm] (g) {
      *& *\\  
		*& *\\ };

\node[above = 1.5cm of g.north] (glabel) {$G$};

\node[
  fit=(g-1-1)(g-1-2.north),
  inner xsep=0,
  above delimiter=\{,
  label=above:$A_2$
 ] {};

 \node[
   fit=(g-1-1)(g-2-1),
   inner xsep=10pt,
   left delimiter=\{,
   label=left:$A_1$
  ] {};

\matrix[right = 2.5cm of g.north east, anchor = north west, matrix of math nodes, left delimiter=(,right delimiter=), row sep=.1cm, column sep=.1cm] (gbar) {
      * & * & * & * & * & *\\ 
      * & * & * & * & * & *\\ 
      * & * & * & * & * & *\\ 
      * & * & * & * & * & *\\ 
      * & * & * & * & * & *\\ 
      * & * & * & * & * & *\\}; 
		
\draw[dashed] ($0.5*(gbar-1-2.north east)+0.5*(gbar-1-3.north west)$) -- ($0.5*(gbar-6-2.south east)+0.5*(gbar-6-3.south west)$);

\draw[dashed] ($0.5*(gbar-2-1.south west)+0.5*(gbar-3-1.north west)$) -- ($0.5*(gbar-2-6.south east)+0.5*(gbar-3-6.north east)$);

\draw[dashed] ($0.5*(gbar-1-4.north east)+0.5*(gbar-1-5.north west)$) -- ($0.5*(gbar-6-4.south east)+0.5*(gbar-6-5.south west)$);

\draw[dashed] ($0.5*(gbar-4-1.south west)+0.5*(gbar-5-1.north west)$) -- ($0.5*(gbar-4-6.south east)+0.5*(gbar-5-6.north east)$);

\node[above = 1.5cm of gbar.north] (gbarlabel) {$\tilde G$};		

\node[
  fit=(gbar-1-1)(gbar-1-2.north),
  inner xsep=0,
  above delimiter=\{,
  label=above:$A_2$
 ] {};
 
\node[
  fit=(gbar-1-3)(gbar-1-4.north),
  inner xsep=0,
  above delimiter=\{,
  label=above:$\hat A_2^{a_2}$
 ] {};

\node[
  fit=(gbar-1-5)(gbar-1-6.north),
  inner xsep=0,
  above delimiter=\{,
  label=above:$\hat A_2^{a_2'}$
 ] {};

\node[
  fit=(gbar-1-1.west)(gbar-2-1.west),
  inner xsep=10pt,
  left delimiter=\{,
  label=left:$A_1$
 ] {};

\node[
  fit=(gbar-3-1.west)(gbar-4-1.west),
  inner xsep=10pt,
  left delimiter=\{,
  label=left:$\hat A_1^{a_1}$
 ] {};
 
\node[
  fit=(gbar-5-1.west)(gbar-6-1.west),
  inner xsep=10pt,
  left delimiter=\{,
  label=left:$\hat A_1^{a_1'}$
 ] {};

\matrix[right = 2.5cm of gbar.north east, anchor = north west, matrix of math nodes,left delimiter=(,right delimiter=), row sep=.1cm, column sep=.1cm] (ghat) {
      * & * & *\\
      * & * & *\\
      * & * & *\\};

\draw[dashed, opacity = 1] ($0.5*(ghat-1-2.north east)+0.5*(ghat-1-3.north west)$) -- ($0.5*(ghat-3-2.south east)+0.5*(ghat-3-3.south west)$);

\draw[dashed, opacity = 1] ($0.5*(ghat-2-1.south west)+0.5*(ghat-3-1.north west)$) -- ($0.5*(ghat-2-3.south east)+0.5*(ghat-3-3.north east)$);

\node[above = 1.5cm of ghat.north] (ghatlabel) {$\hat G$};		

\node[
  fit=(ghat-1-3),
  inner xsep=0,
  label=above:$\hat a_2$
] {};

\node[
  fit=(ghat-1-1)(ghat-1-2.north),
  inner xsep=0,
  above delimiter=\{,
  label=above:$A_2$
 ] {};
 
\node[
  fit=(ghat-3-1),
  inner xsep=10pt,
  label=left:$\hat a_1$
] {};

\node[
  fit=(ghat-1-1)(ghat-2-1),
  inner xsep=10pt,inner ysep=0,
  left delimiter=\{,
  label=left:$A_1$
 ] {};

\draw[->, shorten >=0.3cm, shorten <=.3cm] (glabel) -- (gbarlabel) node[midway,above] {clone actions in $A_i$};
\draw[->, shorten >=0.3cm, shorten <=.3cm] (gbarlabel) -- (ghatlabel) node[midway,above] {permute actions not in $A_i$};
\draw[->, shorten >=0.3cm, shorten <=.3cm] (gbarlabel) -- (ghatlabel) node[midway,below] {and blow down to $\hat a_i$};

\end{tikzpicture} 
\caption{Schematic depiction of the games $G$, $\tilde G$, and $\hat G$ constructed in the proof of \Cref{lem:linearalgebra} (with $n = 2$, $|A_i| =  2$, $k_i = (2,2)$, and $\hat a_i\not\in A_i$ for $i = 1,2$).
$\tilde G$ is obtained from $G$ by adding $k_i(a_i)$ clones of every action $a_i$ of player $i$.
Then, an intermediate game $\bar G$ is constructed from $\tilde G$ by permuting the actions outside of $A_i$ and summing over the resulting games.
The actions outside of $A_i$ are now clones obtained from a convex combination (with weights $k_i$) of actions in $A_i$.
Removing all but one of these clones gives $\hat G$.
}
\label{fig:lem:linearalgebra:proof}
\end{figure}

\begin{proof}
	The first step constructs from $G$ a game $\tilde G$ by adding $k_i(a_i)$ clones of every action (with an exception for $\hat a_i$).
	For all $i\in N$ and $a_i\in A_i$, let $\hat A_i^{a_i}\subseteq U$ so that $|\hat A_i^{a_i}| = k_i(a_i)$ if $a_i\neq \hat a_i$, $|\hat A_i^{\hat a_i}| = k_i(\hat a_i)-1$ if $\hat a_i \in A_i$, and $\hat A_i^{\hat a_i} = \emptyset$ if $\hat a_i\in U\setminus A_i$, and all $\hat A_i^{a_i}$ are disjoint and disjoint from $A_i^-:=A_i\setm\{\hat a_i\}$.
	Let $\tilde A_i = A_i \cup (\bigcup_{a_i\in A_i} \hat A_i^{a_i})$ and $\phi_i\colon \tilde A_i\rightarrow A_i$ so that $\phi_i^{-1}(a_i) = \{a_i\}\cup \hat A_i^{a_i}$.
	Let $\tilde G$ be a game on $\tilde A = \tilde A_1\timesdots \tilde A_n$ so that $\tilde G$ is a blow-up of $G$ with surjection $\phi = (\phi_1,\dots,\phi_n)$.
	Hence, $\tilde G$ is a game with $k_i(a_i)+1$ clones of each action $a_i\in A_i$ if $a_i\neq\hat a_i$ and $k_i(\hat a_i)$ clones of $\hat a_i$.
	Consequentialism implies that $\tilde p\in \phi_*^{-1}(p)\subseteq f(\tilde G)$, where $\tilde p_i = p_i - x_i + \|x_i\| \sum_{a_i\in \tilde A_i\setm A_i^-} \frac{e_{a_i}}{|\tilde A_i\setm A_i^-|}$, which is a lottery since $x_i \le p_i$.
	In words, $\tilde p_i$ is obtained from $p_i$ by subtracting probability $x_i$ from $a_i$ and distributing it uniformly over the added clones of $a_i$, and $a_i$ itself if $a_i = \hat a_i\in A_i$.
	The number of clones has been chosen so that this amounts to subtracting $x_i$ and adding the uniform lottery on $\tilde A_i\setm A_i^-$---the added clones and $\hat a_i$ if $\hat a_i\in A_i$---times $\|x_i\|$.   
	
	The second step constructs from $\tilde G$ as game $\bar G$ by permuting all actions in $\tilde A_i\setminus A_i^-$ in all possible ways and averaging over those permutations of $\tilde G$.  
	Recall that for all $i\in N$, $\Sigma_{\tilde A_i\setm A_i^-} \subseteq \Sigma_{\tilde A_i}$ is the set of all permutations of $\tilde A_i$ that fix each element of $A_i^-$, and let $\Sigma = \Sigma_{\tilde A_1\setm A_1^-}\timesdots\Sigma_{\tilde A_n\setm A_n^-}$.
	Let
	\begin{align*}
		\bar G = \frac1{|\Sigma|}\sum_{\pi\in\Sigma} \tilde G\circ\pi.
	\end{align*}
	By definition, $\bar G$ is invariant under permutations of the actions in $\tilde A_i\setm A_i^-$, and so all actions in $\tilde A_i\setm A_i^-$ are clones of each other.
	Since $\tilde p_i$ assigns the same probability to all actions in $\tilde A_i \setm A_i^-$ and $f$ is equivariant (since it satisfies consequentialism), $\tilde p \in f(\tilde G\circ\pi)$ for all $\pi\in\Sigma$.
	Consistency then implies that $\tilde p\in f(\bar G)$.
	Every action in $\tilde A_i\setm A_i^-$ is by construction the convex combination of actions in $A_i$ with coefficients $k_i$.
	Thus, for all $I\subseteq N$, $a\in A$, and $\tilde a\in (\tilde A_1\setm A_1^-)\timesdots (\tilde A_n\setm A_n^-)$, 
	\begin{align*}
		\bar G(\tilde a_I,a_{-I}) = G\left(\left(\frac{k_i}{\|k_i\|}\right)_{i\in I}, a_{-I}\right).
 	\end{align*}
	Hence, playing $\tilde a_i\in \tilde A_i\setm A_i^-$ in $\bar G$ is payoff-equivalent to playing $\frac{k_i}{\|k_i\|}$ in $G$.
	
	The third step constructs from $\bar G$ a game $\hat G$ by deleting all actions in $\tilde A_i \setm A_i^-$---which are all clones of each other---except for $\hat a_i$.
	For all $i\in N$, let $\hat\phi_i\colon \tilde A_i\rightarrow \hat A_i$ so that $\hat\phi_i$ is the identity $A_i^-$ and $\hat\phi^{-1}(\hat a_i) = \tilde A_i\setm A_i^-$.
	Let $\hat G$ be a blow-down of $\bar G$ with surjection $\hat\phi = (\hat\phi_1,\dots,\hat\phi_n)$.
	Note that $\hat p = \hat\phi_*(\tilde p)$.
	Consequentialism thus gives $\hat p\in f(\hat G)$.
	Moreover, for all $I\subseteq N$ and $a\in A$, 
	\begin{align*}
		\hat G(\hat a_I,a_{-I}) = \bar G(\hat a_I,a_{-I}) = G\left(\left(\frac{k_i}{\|k_i\|}\right)_{i\in I}, a_{-I}\right).
	\end{align*}
\end{proof}

We show that every solution concept satisfying consequentialism, consistency, and rationality returns a \emph{subset} of Nash equilibria by contraposition: we assume that a non-equilibrium profile is returned in some game $G$ and derive a violation of rationality in two steps, each using \Cref{lem:linearalgebra}.
First, we construct from $G$ a game $\bar G$ where a non-equilibrium profile is returned, and every player plays some distinguished action with probability close to 1 with the distinguished action of one player, say $j$, not being a best response. 
Using consequentialism, we may furthermore assume that $j$ plays some best response with probability zero.
The second step is to replace each non-distinguished action of each player, except for $j$'s probability-zero best response, by a convex combination of itself and the distinguished action of that player with large weight on the latter.
In the resulting game $\hat G$, $j$'s probability-zero best response dominates all of $j$'s other actions, violating rationality.
Notably, totality of the solution concept is not used in proving the containment in $\nash$.
The proof that \emph{all} Nash equilibria are returned crucially uses the containment in $\nash$ and totality. It appears in the appendix and shows that every game with equilibrium $p$ can be written as a convex combination of games in which $p$ is the \emph{unique} equilibrium, possibly after removing clones. This statement requires an elaborate proof and may be of independent interest.

\begin{proof}[Proof of \Cref{thm:nash}]
	We prove that $f\subseteq\nash$.
	The proof of $\nash\subseteq f$ is given in the Appendix. 
	
	Let $G$ be a game on $A$ and $p \in f(G)$.
	We assume that the payoffs in $G$ are normalized so that for every player $i$, the maximum and the minimum of $G_i$ differ by at most 1.
	This assumption simplifies the expressions for some bounds below but is inessential. 
	Assume for contradiction that $p\not\in\nash(G)$.
	Then, there is a player $j\in N$ for whom $p_j$ is not a best response.
	That is, $G_j(a_j^*,p_{-j}) - G_j(p) > \eps$ for some $a_j^*\in A_j$ and $\eps > 0$.
	We may assume that $\frac{3n}\eps \in\mathbb N$.
	
	The first step is to replace $G$ by a game $\bar G$, where every player $i$ has an additional action that is payoff-equivalent to playing approximately $p_i$.
	Let $\bar a \in (U\setm A_1)\timesdots (U\setm A_n)$, $k\in(\mathbb Z_+^{U}\setm\{0\})^N$, and $\kappa\in\mathbb R_{+}^N$ so that for all $i\in N$, $\supp(k_i)\subseteq A_i$, $\kappa_i > 0$, and for $x_i = \kappa_i k_i$, $\|x_i\| \ge 1 - \frac{\eps^2}{6n}$ and $x_i \le p_i$.
	By \Cref{lem:linearalgebra}, which applies by the conditions imposed on $k_i$ and $x_i$, there is a game $\bar G$ on $\bar A$ with $\bar A_i =  A_i\cup\{\bar a_i\}$ so that 
	\begin{enumerate}
		\item $\bar p\in f(\bar G)$, where $\bar p_i = p_i - x_i + \|x_i\| e_{\bar a_i}$, and
		\item for all $I\subseteq N$ and $a\in A$, $\bar G(\bar a_I,a_{-I}) = G\left(\left(\frac{k_i}{\|k_i\|}\right)_{i\in I}, a_{-I}\right)$.
	\end{enumerate}
	In particular, $\bar p_j(\bar a_j) \ge 1 - \frac{\eps^2}{6n}$ and, using the normalization of $G$ and the fact that 
	\begin{align*}
		\left\|p_i - \frac{k_i}{\|k_i\|}\right\| = \left\|p_i - \frac{x_i}{\|x_i\|}\right\| \le \|p_i - x_i\| + \left\|x_i - \frac{x_i}{\|x_i\|}\right\| \le \frac{\eps^2}{3n},
	\end{align*}
	we have
	\begin{align*}
		\bar G_j(a_j^*,\bar a_{-j}) - \bar G_j(\bar a) &= G_j\left(a_j^*,\left(\frac{k_i}{\|k_i\|}\right)_{i\neq j}\right) - G_j\left(\left(\frac{k_i}{\|k_i\|}\right)_{i\in N}\right)\\
		&\ge G_j(a_j^*,p_{-j}) - G_j(p) - \sum_{i\in N} \left\|p_i - \frac{k_i}{\|k_i\|}\right\| > \frac{2\eps}3.
	\end{align*}
	
	Using that $f$ satisfies consequentialism, we may add clones of actions, and so we may assume without loss of generality that $\bar p_j(a_j^*) = 0$ and that there is $B\in\mathcal F(U)$ so that $\bar A_j = B \cup \{a_j^*,\bar a_j\}$, and for all $i\neq j$, $\bar A_i = B\cup\{\bar a_i\}$.

	The second step is to modify $\bar G$ so that $a_j^*$ dominates every action in $B\cup\{\bar a_j\}$ by replacing every action in $B$ of every player $i\in N$ by a convex combination of itself and $\bar a_i$ with a sufficiently large weight on $\bar a_i$.
	For every $b\in B$, let $k^b\in\mathbb (\mathbb Z_+^U\setm\{0\})^N$ and $\kappa^b\in\mathbb R_{+}^N$ so that for all $i\in N$ and $x_i^b :=\kappa_i^bk_i^b$, $k_i^b = e_{b} + \frac{3n}\eps  e_{\bar a_i}$ and $x_i^b(b) = \bar p_i(b)$.
	Note that for all $i\in N$,
	\begin{align*}
		\sum_{b\in B} x_i^b(\bar a_i) = \frac{3n}\eps \sum_{b\in B} x_i^b(b) = \frac{3n}\eps \sum_{b\in B} \bar p_i(b)\le \frac{3n}\eps \frac{\eps^2}{6n} = \frac\eps2, 
	\end{align*}
	and hence, $\sum_{b\in B} x_i^b\le \bar p_i$.
	Sequential application of \Cref{lem:linearalgebra} to $\bar G$, one for each $b\in B$ with $\hat a = (b,\dots,b)$, gives a game $\hat G$ on $\bar A$ so that 
	\begin{enumerate}
		\item $\hat p\in f(\hat G)$, where $\hat p_i = \bar p_i - \frac{3n}\eps(1-\|x_i\|) e_{\bar a_i} + \frac{3n}\eps \sum_{b\in B} (p_i(b) - x_i(b)) e_{b}$ for all $i\in N$, and
		\item for all $I\subseteq N$ and $a\in A$, $\hat G(a_I,\bar a_{-I}) = \bar G\left(\left(\frac{k_i^{a_i}}{\|k_i^{a_i}\|}\right)_{i\in I},\bar a_{-I}\right)$.
	\end{enumerate}
	Note that $\hat p_j(a_j^*) = \bar p_j(a_j^*) = 0$ by assumption.
	By the second statement, player $j$ playing $a_j^*$ in $\hat G$ is payoff equivalent to playing $a_j^*$ in $\bar G$, and for all $i\in N$ and $b\in B\cup\{\bar a_i\}$, player $i$ playing $a_i$ in $\hat G$ is payoff equivalent to playing $a_i$ with probability $\frac1{1 + 3n/\eps}$ and $\bar a_i$ with probability $\frac{3n/\eps}{1 + 3n/\eps}$ in $\bar G$.
	Thus, using again the normalization of $G$, we have for all $a\in \bar A$ with $a_j\neq a_j^*$,
	\begin{align*}
		|\hat G_j(a_j^*,a_{-j}) - \bar G_j(a_j^*,\bar a_{-j})|\le \frac n{1 + 3n/\eps} < \frac\eps3 \quad\text{ and }\quad |\hat G_j(a) - \bar G_j(\bar a)| \le \frac\eps3.
	\end{align*}
	It follows that for all $a\in \bar A$ with $a_j\neq a_j^*$,
	\begin{align*}
		\hat G_j(a_j^*,a_{-j}) - \hat G_j(a) \ge \bar G_j(a_j^*,\bar a_{-j}) - \bar G_j(\bar a) - \frac{2\eps}3 > 0.
	\end{align*}
	That is, $a_j^*$ dominates every other action of $j$ in $\hat G$, and so is dominant.
	Since $\hat p_j$ assigns probability~0 to $a_j^*$, this contradicts rationality.
\end{proof}

\section{Robustness of the Characterization}\label{sec:robustness}

In many areas of mathematics, it is common to aim for robust versions of results.
In this spirit, we show that the inclusion $f\subseteq\nash$ in \Cref{thm:nash} is robust with respect to small violations of the axioms: every solution concept that approximately satisfies the three axioms is approximately Nash equilibrium.
This can be made precise by formulating quantitatively relaxed versions of the axioms and replacing Nash equilibrium with $\eps$-equilibrium.

$\eps$-equilibrium is the standard notion of an approximate equilibrium.
A strategy profile is an $\eps$-equilibrium if no player can deviate to a strategy that increases her payoff by more than $\eps$.
By $\nash_\eps$ we denote the solution concept returning all $\eps$-equilibria in all games.

\begin{definition}[$\eps$-equilibrium]\label{def:approxnash}
	Let $G$ be a game on $A = A_1\timesdots A_n$.
	A profile $p$ is an $\eps$-equilibrium of $G$ if
	\begin{align*}
		G_i(p_i,p_{-i}) \ge G_i(q_i,p_{-i}) - \eps\text{ for all } q_i \in \Delta A_i\text{ and }i\in N.
	\end{align*}
\end{definition}

Likewise, there are natural approximate notions of consequentialism, consistency, and rationality.
They are obtained from the exact versions defined in the preceding section by allowing for small perturbations of strategy profiles.
 
Recall that consequentialism requires that if $G$ is a blow-up of $G'$, then a profile is returned in $G$ if and only if its pushforward is returned in $G'$.
Approximate consequentialism weakens this condition by requiring only that the set of returned profiles for $G$ is close (in Hausdorff distance) to the set of profiles whose pushforward is returned in $G'$.
\begin{definition}[$\delta$-consequentialism]\label{def:approxconseq}
	A solution concept $f$ satisfies $\delta$-consequentialism if for all games $G$ and $G'$ such that $G$ is a blow-up of $G'$ with surjection $\phi = (\phi_1,\dots,\phi_n)$, $f(G) = \phi_*^{-1}(\phi_*(f(G)))$, and
	\begin{align*}
		f(G)\subseteq B_\delta \left(\phi_*^{-1}(f(G'))\right) \text { and } \phi_*^{-1}(f(G')) \subseteq B_\delta(f(G)).
	\end{align*}
\end{definition}
The first assertion requires that probability can be distributed arbitrarily among clones.
The first set inclusion asserts that in the game obtained by cloning actions, the solution concept can only return profiles that differ by no more than $\delta$ from some profile obtained as a blow-up of a profile that is returned in the original game.
Conversely, the second inclusion requires that every blow-up of a profile returned in the original game differs by at most $\delta$ from some profile returned in the blown-up game.

Approximate consistency weakens its exact counterpart by requiring only that if a profile is returned in several games, then some profile close to it has to be returned in any convex combination of these games.

\begin{definition}[$\delta$-consistency]\label{def:approxconsistency}
	A solution concept $f$ satisfies $\delta$-consistency if for all games $G^1,\dots,G^k$ on the same action profiles and every $\lambda\in\mathbb R_+^k$ with $\sum_j\lambda_j = 1$,
	\begin{align*}
		f(G^1)\cap\dots\cap f(G^k)\subseteq B_{\delta}\left(f\left(\lambda_1 G^1 + \dots + \lambda_kG^k\right)\right).
	\end{align*}
\end{definition}
Unlike the exact notion, $\delta$-consistency as defined is not equivalent to its restriction for $k = 2$.\footnote{A proof attempt by induction fails for two reasons.
First, every application of $\delta$-consistency introduces an additive error of $\delta$, so that $k-1$ applications to two games only gives an error bound of $(k-1)\delta$ on the right hand side.
Second, even if $f(G_1)$, $f(G_2)$, and $f(G_3)$ have a non-empty common intersection, $f(\frac{\lambda_1}{\lambda_1 + \lambda_2} G_1 + \frac{\lambda_2}{\lambda_1 + \lambda_2} G_2)$ need not intersect with $f(G_3)$, making a further application of $\delta$-consistency useless.}

Lastly, approximate rationality asserts that actions that are dominated by a non-negligible amount are not played too frequently, say, with probability at most $\nicefrac12$.\footnote{The bound of $\nicefrac12$ could be replaced by any bound strictly below $1$.}
If $G$ is a game on $A= A_1\timesdots A_n$ and $a_i,a_i'\in A_i$ are actions of player $i$, we say that $a_i$ \emph{$\delta$-dominates} $a_i'$ if $G(a_i,a_{-i}) \ge G(a_i',a_{-i}) + \delta$ for all $a_{-i}\in A_{-i}$. 

\begin{definition}[$\delta$-rationality]\label{def:approxrationality}
	A solution concept $f$ satisfies $\delta$-rationality if for all games $G$,
	\begin{align*}
		f(G) \subseteq B_{\nicefrac12}(\hat A_1^\delta)\timesdots B_{\nicefrac12}(\hat A_n^\delta),
	\end{align*}
	where $\hat A_i^\delta$ denotes the set of actions of player $i\in N$ that are not $\delta$-dominated in $G$.
\end{definition}

Note that $\delta$-consequentialism and $\delta$-consistency reduce to the exact versions defined in the previous section when $\delta = 0$, and $\delta$-rationality is slighly stronger than its non-approximate counterpart.
We call a solution concept \emph{$\delta$-nice} if it satisfies $\delta$-consequentialism, $\delta$-consistency, and $\delta$-rationality.
It is routine to check that $\nash_\eps$ is $\delta$-nice for small enough $\delta$. 
Conversely, we show that for small enough $\delta$, every $\delta$-nice solution concept is a refinement of $\nash_\eps$.  
The statement is restricted to normalized games and requires equivariance for reasons that we discuss in Remarks \ref{rem:normalized} and \ref{rem:equivariance}.

\begin{theorem}\label{thm:continuityofcharacterization}
	Consider solution concepts on the set of normalized games.
	Then, for every $\eps> 0$, there is $\delta>0$ so that if $f$ is equivariant and satisfies $\delta$-consequentialism, $\delta$-consistency, and $\delta$-rationality, then $f$ is a refinement of $\nash_\eps$.
\end{theorem}

Note that $\delta$ in \Cref{thm:continuityofcharacterization} does not depend on $f$.
Otherwise, the statement would follow from the fact that exact consequentialism, consistency, and rationality characterize Nash equilibrium.
The proof is similar to that of \Cref{thm:nash}.
However, apart from the need to keep track of error terms arising from applications of the axioms, some steps require additional care.
The proof appears in \Cref{sec:robustnessproof}.
The comments on subclasses of games in \Cref{sec:discussion} remain valid for \Cref{thm:continuityofcharacterization}.

\begin{remark}[Normalized games]\label{rem:normalized}
	\Cref{thm:continuityofcharacterization} is restricted to normalized games since $\eps$-equilibrium becomes too stringent without a bound on the payoffs.
	More precisely, for any game $G$, the set of $\eps$-equilibria of $cG$ shrinks to the set of exact equilibria of $G$ as $c$ goes to positive infinity. 
	To see that the restriction to normalized games is indeed necessary, consider the following solution concept $f$ for two-player games. 
	Let $\hat G = (0,0\quad 0,-c)$ be a game where the first player has only one action, the second player has two actions, and $c$ is a large positive number ($c\gg \frac1\delta$), and let $\hat p = (1,(1-\delta,\delta))$.
	Observe that $\hat p$ is not an $\eps$-equilibrium for $\hat G$.
	Define $f(G) = \nash(G) \cup \{\hat p\}$ if $G = \hat G$ and $f(G) = \nash(G)$ otherwise.
	It is not hard to see that $f$ is $\delta$-nice.\footnote{$\delta$-consequentialism and $\delta$-rationality follow from the fact that $\nash$ is nice and the definition of the axioms.
	To verify $\delta$-consistency, it suffices to consider convex combinations of games with equilibrium $\hat p$ involving $\hat G$.
	The only games on the same action sets as $\hat G$ with equilibrium $\hat p$ are those where both actions of the second player give her the same payoff.
	Any convex combination of such games with $\hat G$ has $(1,(1,0))$ as the unique equilibrium, which conforms with $\delta$-consistency.}
\end{remark}

\begin{remark}[Equivariance]\label{rem:equivariance}
	Applying $\delta$-consequentialism to two games that are the same up to permuting actions, one can see that it implies $\delta$-equivariance defined analogously to $\delta$-consequentialism. 
	However, our proof requires exact equivariance.
	Whether \Cref{thm:continuityofcharacterization} holds without this assumption is open.
\end{remark}

\begin{remark}[Converse of \Cref{thm:continuityofcharacterization}]
	In contrast to \Cref{thm:nash}, a characterization of $\nash_\eps$ as the only $\delta$-nice solution for some $\delta$ is not possible.
	If $\nash_\eps$ is $\delta$-nice for some $\delta$, then so is $\nash_{\eps'}$ for all $\eps' \le \eps$.
	A weaker converse to \Cref{thm:continuityofcharacterization} would require that for every $\delta > 0$, there is $\eps > 0$ so that if $f$ is $\delta$-nice, then $f = \nash_{\eps'}$ for some $\eps' \le \eps$.
	While this statement is obviously false in general since $\delta$-niceness is vacuous for large $\delta$, we do not know if it holds for small enough $\delta$.
\end{remark}

\begin{remark}[Alternative approximate equilibrium notions]\label{rem:fragilenash}
	\Cref{thm:continuityofcharacterization} fails if $\eps$-equilibrium is replaced by some alternative notions of approximate equilibrium.
	For example, \Cref{thm:continuityofcharacterization} does not hold when replacing $\nash_\eps$ by $B_\eps(\nash)$, the set of profiles that are $\eps$-close to some Nash equilibria.
	This is because $\nash_{\eps'}$ is $\delta$-nice for small enough $\eps'$, however, $\nash_{\eps'} \not\subseteq B_\eps(\nash)$ for any $\eps' > 0$.
	In words, no matter how small $\eps'$ is, there exist $\eps'$-equilibria that are more than $\eps$ away from every Nash equilibrium.
\end{remark}

\section{Discussion}\label{sec:discussion}

We conclude by discussing consequences and variations of our characterization of Nash equilibrium, as well as the independence of the axioms. 

\paragraph{Rationalizability and admissibility}
A player's strategy is rationalizable if it is a best response to some belief about the other players' strategies, assuming that everyone's rationality is common knowledge \citep{Bern84a,Pear84a}.
Every Nash equilibrium strategy is rationalizable; thus, rationalizability is consistent with our axioms.
However, the solution concept returning all rationalizable strategy profiles---profiles in which every player plays a rationalizable strategy---violates consistency. 
In the first two games below, every strategy of either player is rationalizable.\footnote{In two-player games, an action is rationalizable if and only if it survives the iterated elimination of strictly dominated strategies \citep{Bern84a,Pear84a}. Since no strategy is strictly dominated in either game, every strategy is rationalizable in both games.}
In the first game, the second action of the row player is rationalizable by a belief with high probability on the column player's first action, whereas in the second game, the second action of the row player is rationalizable by a belief with high probability on the column player's second action.
The third game is the uniform convex combination of the first two.
In that game, the second action of the row player is strictly dominated and thus not rationalizable, which shows the violation of consistency.
These consistency failures arise because the strategy profile is not common knowledge: a player may rationalize a strategy with a belief that differs from the other players' strategies. 
However, note that common knowledge of the strategy profile and the rationality of all players already entail that the strategy profile is a Nash equilibrium.
	\begin{center}
	\begin{tikzpicture}[baseline,
  label distance=10pt]
		\matrix [matrix of math nodes, left delimiter=(,right delimiter=), row sep=.1cm, column sep=.1cm] (g) {
		      0,0 & 4,0 \\ 
				2,0 & 0,0 \\ };
		\matrix [right = 3.4cm of g.north, anchor = north west, matrix of math nodes, left delimiter=(,right delimiter=), row sep=.1cm, column sep=.1cm] (g2) {
		      4,0 & 0,0 \\ 
				0,0 & 2,0 \\ };
		\matrix [right = 3.4cm of g2.north, anchor = north west, matrix of math nodes, left delimiter=(,right delimiter=), row sep=.1cm, column sep=.1cm] (g3) {
		      2,0 & 2,0 \\ 
				1,0 & 1,0 \\ };
		\node[fit=(g),inner xsep=0pt,label=left:$\nicefrac12$] {};
		\node[fit=(g2),inner xsep=0pt,label=left:$+\;\nicefrac12$] {};
		\node[fit=(g3),inner xsep=0pt,label=left:${=}$] {};
	\end{tikzpicture}
	\end{center}
	
A strategy is admissible if it is not weakly dominated or, equivalently, if it is a best response to some belief about the other players' strategies with full support. 
While each of our axioms individually is compatible with admissibility,\footnote{The solution concept returning all admissible strategy profiles satisfies consequentialism and rationality. The solution concept returning all strategy profiles in which every player best responds to uniformly randomizing opponents satisfies consistency (see also the discussion of the independence of the axioms below), and each such strategy profile is admissible by the stated equivalence.} their conjunction is not since not all Nash equilibria are admissible.
The example for rationalizability above also shows that the solution concept returning all admissible strategy profiles violates consistency.

\paragraph{Equilibrium refinements}
	One consequence of \Cref{thm:nash} is that every refinement of Nash equilibrium violates at least one of the axioms (including totality). 
	We discuss some examples.
	Since rationality is preserved under taking subsets, every refinement satisfies rationality.
	\emph{Quasi-strict equilibrium} \citep{Hars73a} satisfies consistency, and, for two players, is total \citep{Nord99a}.
	However, it violates consequentialism (even for two players) since it does not allow for the possibility that clones of equilibrium actions are played with probability~0.
	A trivial example is a game where all players' utility functions are constant for all action profiles.
	Then, every \emph{full support} strategy profile is a quasi-strict equilibrium, whereas consequentialism requires that every strategy profile is returned.
	For three or more players, quasi-strict equilibria may not exist.
	
	\emph{Trembling-hand perfect equilibrium} \citep{Selt75a} is total and satisfies consequentialism; thus, it is not consistent.
	In the first two games below, the strategy profile with probability 1 on the bottom-right action profile is a trembling-hand perfect equilibrium.
	To see this, recall that in two-player games, an equilibrium is trembling-hand perfect if and only if it is admissible \citep[see, e.g.,][]{vDam91a}. 
	The third game is the uniform convex combination of the first two.
	In that game, the row player's second action is weakly dominated; thus, the same equilibrium is not trembling-hand perfect, which shows a consistency violation. 
	The reason is the same as for the consistency violation discussed in the previous paragraph: in the first game, the second action of the row player is justified by trembles with much more probability on the column player's first action, whereas in the second game, the second action of the row player is justified by trembles with much more probability on the column player's second action.   
	\begin{center}
	\begin{tikzpicture}[baseline,
  label distance=10pt]
		\matrix [matrix of math nodes, left delimiter=(,right delimiter=), row sep=.1cm, column sep=.1cm] (g) {
		      0,0 & 4,0 & 2,0\\ 
				2,0 & 0,0 & \mathbf{2,0}\\ };
		\matrix [right = 3.4cm of g.north, anchor = north west, matrix of math nodes, left delimiter=(,right delimiter=), row sep=.1cm, column sep=.1cm] (g2) {
		      4,0 & 0,0 & 2,0\\ 
				0,0 & 2,0 & \mathbf{2,0}\\ };
		\matrix [right = 3.4cm of g2.north, anchor = north west, matrix of math nodes, left delimiter=(,right delimiter=), row sep=.1cm, column sep=.1cm] (g3) {
		      2,0 & 2,0 & 2,0\\ 
				1,0 & 1,0 & \mathbf{2,0}\\ };
		\node[fit=(g),inner xsep=0pt,label=left:$\nicefrac12$] {};
		\node[fit=(g2),inner xsep=0pt,label=left:$+\;\nicefrac12$] {};
		\node[fit=(g3),inner xsep=0pt,label=left:${=}$] {};
	\end{tikzpicture}
	\end{center}

	Lastly, \emph{strong equilibrium} \citep{Auma59a} and \emph{coalition-proof equilibrium} \citep{BPW87a} also satisfy consequentialism but violate consistency. To see this, consider the example below. In the first two games, the strategy profile with probability 1 on the bottom-right action profile is a strong (and thereby also coalition-proof) equilibrium. In the third game, which is a convex combination of the first two, the bottom-right action profile is not coalition-proof (and thereby not strong). 

	\begin{center}
	\begin{tikzpicture}[baseline,
  label distance=10pt]
		\matrix [matrix of math nodes, left delimiter=(,right delimiter=), row sep=.1cm, column sep=.1cm] (g) {
		      4,0 & 0,0 \\ 
				0,0 & \mathbf{1,1} \\ };
		\matrix [right = 3.4cm of g.north, anchor = north west, matrix of math nodes, left delimiter=(,right delimiter=), row sep=.1cm, column sep=.1cm] (g2) {
		      0,4 & 0,0 \\ 
				0,0 & \mathbf{1,1} \\ };
		\matrix [right = 3.4cm of g2.north, anchor = north west, matrix of math nodes, left delimiter=(,right delimiter=), row sep=.1cm, column sep=.1cm] (g3) {
		      2,2 & 0,0 \\ 
				0,0 & \mathbf{1,1} \\ };
		\node[fit=(g),inner xsep=0pt,label=left:$\nicefrac12$] {};
		\node[fit=(g2),inner xsep=0pt,label=left:$+\;\nicefrac12$] {};
		\node[fit=(g3),inner xsep=0pt,label=left:${=}$] {};
	\end{tikzpicture}
	\end{center}
	Moreover, strong equilibrium and coalition-proof equilibrium are not total even when there only two players.

\paragraph{Independence of axioms}
	All properties in \Cref{thm:nash} are required to derive the conclusion.
	For each of the four axioms (including totality), there is a solution concept different from \nash that satisfies the three remaining axioms.
	\begin{enumerate}[leftmargin=*,label=\textit{(\roman*)}]
		\item \emph{Consequentialism:} return all strategy profiles in which every player randomizes only over actions that are best responses against uniformly randomizing opponents;
		satisfies totality, consistency, and rationality but violates consequentialism.
		
		\item \emph{Consistency:} return all strategy profiles in which every player randomizes only over actions that maximize this player's highest possible payoff;
		satisfies totality, consequentialism, and rationality but violates consistency.
		
		\item \emph{Rationality:} return all strategy profiles that maximize the sum of all players' payoffs;
		satisfies totality, consequentialism, and consistency but violates rationality.
		\item \emph{Totality:} return all strategy profiles whose pushforwards are pure Nash equilibria in a blowdown of the original game; satisfies consequentialism, rationality, and consistency but violates totality.
	\end{enumerate}
	The first three examples are neither contained in nor contain $\nash$ and, noting that they are equivariant, also apply to \Cref{thm:continuityofcharacterization}. The last one is necessarily a refinement of \nash since totality is not needed for the inclusion $f\subseteq \nash$ in \Cref{thm:nash}.
	Further examples that \emph{are} refinements or coarsenings of $\nash$ are not hard to find. 
	Quasi-strict equilibrium (for two players) violates consequentialism but satisfies consistency and rationality as discussed above.
	Trembling-hand perfect equilibrium is not consistent but satisfies the other two axioms.
	The trivial solution concept returning all strategy profiles in all games violates rationality but satisfies the remaining two axioms.
	Note that rationality is so weak that even for \emph{one}-player games, all three axioms are required for the characterization.
		
\paragraph{Restricted classes of games}
	Examining the proof of the inclusion $f\subseteq\nash$, one can see that it remains valid for any class of games that is closed under blowing-up, blowing-down, and taking convex combinations.
	More precisely, it holds for any class of games $\mathcal G$ with the following properties.
	\begin{enumerate}[leftmargin=*,label=\textit{(\roman*)}]
		\item If $G$ is a blow-up of $G'$, then $G\in\mathcal G$ if and only if $G'\in\mathcal G$.
		\item If $G_1,\dots,G_k\in\mathcal G$ are games on the same action profiles and $\lambda\in\mathbb R_+^k$ with $\sum_j \lambda_j = 1$, then $\lambda_1 G_1 + \dots + \lambda_k G_k\in \mathcal G$.
	\end{enumerate}
	Similarly, \Cref{thm:continuityofcharacterization} remains valid any such class of games without changes to the proof.
	
	Various well-known classes of games satisfy these properties, for example, (strategically) zero-sum games, graphical games, and potential games. A game is symmetric if all players have the same set of actions, and permuting the actions in any action profile results in the same permutation of the players' payoffs. Symmetric games are not closed with respect to blow-ups.
	For example, cloning or permuting actions of only one player makes a symmetric game asymmetric. 
	It is unclear how to extend the current proof approach to symmetric games. The following symmetric two-player game illustrates this.
	
\begin{center}
\begin{tikzpicture}[baseline]

\matrix [matrix of math nodes, left delimiter=(,right delimiter=), row sep=.1cm, column sep=.1cm] (g) {
        3,3 & 2,2 & 2,2\\  
		2,2 & 3,3 & 0,0\\ 
		2,2 & 0,0 & 3,3\\};
\node[
  fit=(g-1-1),
  inner xsep=0,
  label=above:$0$
 ] {};
\node[
  fit=(g-1-2),
  inner xsep=0,
  label=above:$\nicefrac12$
 ] {};
\node[
  fit=(g-1-3),
  inner xsep=0,
  label=above:$\nicefrac12$
 ] {};
 
\node[
  fit=(g-1-1),
  inner xsep=13pt,
  label=left:$0$
 ] {};
\node[
  fit=(g-2-1),
  inner xsep=13pt,
  label=left:$\nicefrac12$
 ] {};
\node[
  fit=(g-3-1),
  inner xsep=13pt,
  label=left:$\nicefrac12$
 ] {};
\end{tikzpicture} 
\end{center}
	The indicated strategy profile is not a Nash equilibrium.  
	However, all symmetric games that are blow-ups of this game and convex combinations thereof have the same payoffs on the diagonal.
	Hence, clones of the second and third action of either player are not dominated in any of these games, so no contradictions to rationality occur.
	This issue does not arise for symmetric \emph{zero-sum} games \citep[see][Remark 3]{BrBr17c}.

	For our proof of the converse inclusion, $\nash\subseteq f$, a class of games needs to have enough games with a unique equilibrium.
	Suppressing technicalities, it is required that for every game $G\in\mathcal G$ and every equilibrium $p$ of $G$, $G$ can be written as a convex combination of games in $\mathcal G$ that have $p$ as the unique equilibrium.
	We have not examined which classes of games, other than the class of all games, have this property.

\paragraph{Closures of solution concepts}
	One can ``repair'' any given solution concept by iteratively adding strategy profiles whenever there is a failure of consequentialism or consistency.
	For instance, if a profile is returned in two games but not in some convex combination thereof, it is added to the set of returned profiles for the convex combination to eliminate this failure of consistency.
	One can equivalently define the closure of a solution concept $f$ as the smallest solution concept containing $f$ and satisfying consequentialism and consistency.\footnote{This closure is well-defined since consequentialism and consistency are preserved under arbitrary intersections of solution concepts.}
	By \Cref{thm:nash}, the closure of a total refinement of Nash equilibrium is Nash equilibrium, and the closure of total non-refinements violates rationality.

\paragraph{Correlated equilibrium}
Our framework excludes correlated equilibrium since a strategy profile consists of a strategy for each player rather than a distribution over action profiles.
It is an intriguing question whether our results extend to solution concepts returning correlated strategy profiles, but there are several obstacles to obtaining such a result.
First, the axioms must be defined for correlated solution concepts, which leads to subtle issues in the case of consequentialism.
Second, new proof techniques will be required since our arguments crucially exploit that the players' strategies are independent.
Third, Nash equilibrium will satisfy many reasonable extensions of the axioms to correlated solution concepts, and thus, a unique characterization of correlated equilibrium without further axioms may not be feasible.

\section*{Acknowledgments}{\footnotesize%
Florian Brandl acknowledges support by the DFG under the Excellence Strategy EXC-2047. Felix Brandt acknowledges support by the DFG under grants {BR~2312/11-2} and {BR~2312/12-1}.
The authors thank Francesc Dilm\'e, Benny Moldovanu, and Lucas Pahl for helpful feedback.
A preliminary version of this paper was presented at the Interdisciplinary CIREQ-Workshop at Universit\'e de Montr\'eal (Montr\'eal, March 2023), the Microeconomic Theory Workshop at the University of Bonn (Bonn, May 2023), and the Conference on Voting Theory and Preference Aggregation (Karlsruhe, October 2023).\par}

\pagebreak
\appendix
\section*{APPENDIX}\label{sec:appendix}

 \section{Omitted Proof From \Cref{sec:nashcharacterization}}\label{sec:nashcharacterizationproofs}

As a shorthand, we say that a solution concept is \emph{nice} if it satisfies consequentialism, consistency, and rationality. This appendix contains the proof of the missing direction of \Cref{thm:nash}, that is, $\nash\subseteq f$ for any nice solution concept $f$.
The main idea of the proof is simple: for every game $G$ and every equilibrium $p$ of $G$, show that $G$ can be written as a convex combination of games for which $p$ is the unique equilibrium.
Since $f\subseteq\nash$ and $f$ is total, we know that $f$ has to return unique equilibria.
Consistency thus gives $p\in f(G)$.

The difficult part is to find a suitable representation of $G$ as a convex combination.
A first observation is that it suffices to prove that $\nash\subseteq f$ holds for games where the payoff functions of all players but one are 0.
More formally, we say that $G$ is a player $i$ payoff game if for all $j\neq i$, $G_j\equiv 0$.
Then, the following holds.

\begin{lemma}[Reduction to player $i$ payoff games]\label{lem:playeripayoff}
	Let $G$ be a game and $p\in\nash(G)$.
	For $i\in N$, let $G^i$ be the game with $G^i_i = G_i$ and $G^i_j \equiv 0$ for all $j\neq i$.
	Then, $p\in\nash(G^i)$.
\end{lemma}
\begin{proof}
	First, $p_i$ is a best response to $p_{-i}$ in $G^i$ since it is a best response in $G$ and $G^i_i = G_i$.
	Second, for all $j\neq i$, $p_j$ (and any other strategy for that matter) is a best response to $p_{-j}$ in $G^i$ since $G^i_j \equiv 0$.
	Hence, $p\in\nash(G)$.
\end{proof}

So if we can show that for every player $i$ payoff game $G$, $\nash(G)\subseteq f(G)$, we can use \Cref{lem:playeripayoff} and consistency to conclude that the same conclusion holds for all games.
While this reduction is convenient, it is not as powerful as it may seem since player $i$ payoff games have a unique equilibrium only if all players other than $i$ have only a single action.
Thus, even when decomposing player $i$ payoff games into games with a unique equilibrium, one needs to consider games with non-zero payoffs for all players.

\subsection{Reduction to Deterministic Slice-Stochastic Tensors}\label{sec:reductiontoslicestochastic}

The next step is a further reduction showing that it is sufficient to consider the case when $G_i$ is a \emph{slice-stochastic tensor}.
To motivate this notion, recall that the well-known Birkhoff-von Neumann theorem states that every bistochastic matrix can be written as a convex combination of permutation matrices \citep{Birk46a,vNeu53a}.\footnote{A matrix $M\in\mathbb R_+^{m\times m}$ is bistochastic if the row sums and column sums are 1.}
There are different ways one might try to generalize this statement to higher-order tensors.
For example, one might say that a tensor $T\colon A_1\timesdots A_n\rightarrow\mathbb R_+$ is $n$-stochastic if for all $i\in N$ and $a_{-i}\in A_{-i}$, $\sum_{a_i\in A_i} T(a_i,a_{-i}) = 1$ (which is to say that every ``tube'' of $T$ sums to 1).
However, with this definition, for $n \ge 3$, it is not true that every $n$-stochastic tensor can be written as a convex combination of $n$-stochastic tensors taking values in $\{0,1\}$ \citep[see][]{CLN14a}.
We thus opt for a different generalization of bistochastic matrices. 

\begin{definition}[Slice-stochastic tensors]
	Let $A\in\mathcal F(U)^n$ with $|A_1| = \dots = |A_n|$.
	A tensor $T\colon A\rightarrow\mathbb R$ is \emph{slice-stochastic for $i\in N$} if
	\begin{enumerate}[leftmargin=*,label=\textit{(\roman*)}]
		\item for all  $a_{-i}\in A_{-i}$, $\sum_{a_i\in A_i} T(a_i,a_{-i}) = 1$, \label{item:slice1}
		\item for all $a_i\in A_i$, $\sum_{a_{-i}\in A_{-i}} T(a_i,a_{-i}) = m^{n-2}$, and\label{item:slice2}
		\item for all $a\in A$, $0 \le T(a) \le 1$.\label{item:slice3}
	\end{enumerate}
	We say that $T$ is a \emph{deterministic} slice-stochastic tensor if it is slice-stochastic and takes values in $\{0,1\}$.
\end{definition}

For $n = 2$, $T$ is a bistochastic matrix if and only if it is slice-stochastic for some $i = 1,2$.
Note that if $T$ is slice-stochastic for $i$, then
\begin{align*}
	\sum_{a\in A} T(a) = \sum_{a_{-i}\in A_{-i}} 1 = \sum_{a_i\in A_i} m^{n-2} = m^{n-1}.
\end{align*}
We omit writing ``for $i$'' when $i$ is clear from the context.

It turns out that the Birkhoff-von Neumann theorem does extend to slice-stochastic tensors of any order.
That is, every slice-stochastic tensor is a convex combination of deterministic slice-stochastic tensors.
This will allow us to reduce the problem $\nash\subseteq f$ to payoff functions of the latter type.

\begin{lemma}[Birkhoff-von Neumann theorem for slice-stochastic tensors]\label{lem:bvnslice}
	Let $A\in\mathcal F(U)^n$ with $|A_1| = \dots = |A_n|$.
	Let $T\colon A\rightarrow\mathbb R$ be a slice-stochastic tensor for $i\in N$.
	Then, there are tensors $T^1,\dots,T^K\colon A\rightarrow \{0,1\}$ that are slice-stochastic for $i$ and $(\lambda^1,\dots,\lambda^K)\in\Delta([K])$ so that
	\begin{align*}
		T = \sum_{k\in[K]} \lambda^k T^k.
	\end{align*}
\end{lemma}
\begin{proof}
	Viewing $T$ as an element of $\mathbb R^A$, $T$ is slice-stochastic for $i$ if it is a solution to the linear feasibility program 
	\begin{align*}
		 Mx\le v,
	\end{align*}
	where the matrix $M$ and the vector $v$ are given by the constraints of type~\ref{item:slice1},~\ref{item:slice2}, and~\ref{item:slice3}.
	Thus, $M$ has $2|A_{-i}| + 2|A_i| + 2|A| = 2(m^{n-1} + m + m^n)$ rows (2 for each constraint of each of the three types) and $|A| = m^n$ columns; the number of rows of $v$ is the same as for $M$.
	Note that $M$ is of the form $M = (\tilde M^\intercal,-\tilde M^\intercal)^\intercal$ for some matrix $\tilde M$, since each constraint gives two rows in $M$ where one is the negative of the other.
	(For $n = 3$ and $m = 2$, the matrix $\tilde M$ is depicted in \Cref{fig:mat}.)
	
	We want to show that the polytope defined by $Mx\le v$ has integral vertices. 
	Since $v$ is integral (in fact, $\{-1,0,1\}$), it suffices to show that $M$ is totally unimodular.\footnote{$M$ is totally unimodular if every square submatrix of $M$ has determinant $-1$, $0$, or $1$.}
	But $M$ is totally unimodular if and only if $\tilde M$ is totally unimodular.
	Now $\tilde M$ is totally unimodular if for every subset $R$ of rows of $\tilde M$, there is an assignment $\sigma\colon R\rightarrow \{-1,1\}$ of signs to the rows in $R$ so that for all $a\in A$,
	\begin{align}
		\sum_{r\in R} \sigma(r) \tilde M_{r,a} \in\{-1,0,1\}.\label{eq:rowsum}
	\end{align}
	A proof of this result appears, for example, in the book by \citet{Schr98a}.
	 
	This condition is easy to check in the present case.
	Let $R$ be a subset of rows of $\tilde M$.
	It is easy to see that we may assume that $R$ does not contain rows corresponding to constraints of type \ref{item:slice3} since those rows only contain of single 1 or $-1$ and can thus always be signed so as to not introduce violations of \eqref{eq:rowsum}.
	We then define $\sigma$ as follows.
	\begin{itemize}
		\item For each $r\in R$ corresponding to a constraint of type \ref{item:slice1} for $a_{-i}\in A_{-i}$, let $\sigma(r) = 1$.
		\item For each $r\in R$ corresponding to a constraint of type \ref{item:slice2} for $a_{i}\in A_{i}$, let $\sigma(r) = -1$.
	\end{itemize}
	Then, for each column index $a\in A$, there is a most one 1 and at most one $-1$ in the sum in \eqref{eq:rowsum}, which concludes the proof.
\begin{figure}
 	\centering
	\begin{tikzpicture}[
  baseline,
  label distance=10pt 
]

\matrix [matrix of math nodes, left delimiter=(,right delimiter=), row sep=.1cm, column sep=.1cm] (g) {
      1 & 1 & \phantom{1} & \phantom{1} & \phantom{1} & \phantom{1} & \phantom{1} & \phantom{1}\\ 
		&& 1 & 1 & \phantom{1} & \phantom{1} & \phantom{1} & \phantom{1} \\ 
		&&&& 1 & 1 & \phantom{1} & \phantom{1} \\ 
		\phantom{1} & & & & & & 1 & 1 \\
		1 && 1 && 1 && 1 & \phantom{1}\\
		\phantom{1} & 1 & \phantom{1} & 1 & \phantom{1} & 1 & \phantom{1} & 1\\
		\phantom{1} & \phantom{1} & \phantom{1} & \phantom{1} & \phantom{1} & \phantom{1} & \phantom{1} & \phantom{1}\\};

 \node[fit=(g-7-4)(g-7-5)]{$I_8$};

 \node[
   fit=(g-1-1.center)(g-4-1.center),
   inner xsep=17pt,
   left delimiter=\{,
   label=left:type~\ref{item:slice1}
  ] {};
  
  \node[
    fit=(g-5-1.center)(g-6-1.center),
    inner xsep=17pt,
    left delimiter=\{,
    label=left:type~\ref{item:slice2}
   ] {};
	
  \node[
    fit=(g-7-1.center),
    inner xsep=17pt,
    left delimiter=\{,
    label=left:type~\ref{item:slice3}
   ] {};

\draw[dashed] ($0.5*(g-1-2.north east)+0.5*(g-1-3.north west)$) -- ($0.5*(g-6-2.south east)+0.5*(g-6-3.south west)$);

\draw[dashed] ($0.5*(g-1-4.north east)+0.5*(g-1-5.north west)$) -- ($0.5*(g-6-4.south east)+0.5*(g-6-5.south west)$);

\draw[dashed] ($0.5*(g-1-6.north east)+0.5*(g-1-7.north west)$) -- ($0.5*(g-6-6.south east)+0.5*(g-6-7.south west)$);

\draw[dashed] ($0.5*(g-4-1.south west)+0.5*(g-5-1.north west)$) -- ($0.5*(g-4-8.south east)+0.5*(g-5-8.north east)$);

\draw[dashed] ($0.5*(g-6-1.south west)+0.5*(g-7-1.north west)$) -- ($0.5*(g-6-8.south east)+0.5*(g-7-8.north east)$);

\end{tikzpicture} 
\caption{The matrix $\tilde M$ in the proof of \Cref{lem:bvnslice} for $n = 3$ and $m = 2$. $I_8$ denotes the identity matrix with $8 = m^n$ rows and columns. Each of the four pairs of columns separated by dashed lines corresponds to some $a_{-i}\in A_{-i}$. One can check that $\tilde M$ is totally unimodular by verifying the equivalent condition~\eqref{eq:rowsum}.}
\label{fig:mat}
\end{figure}
\end{proof}

To show that considering games where every player's payoff function is slice-stochastic games is sufficient, we examine how certain changes to the payoff functions influence the set of equilibria.
For $\alpha\in\mathbb R_{++}^N$ and $\beta\in \mathbb R^N$, we write $\alpha G + \beta$ for the game with payoff function $(\alpha G + \beta)_i = \alpha_i G_i + \beta_i$ for all $i\in N$.
Moreover, we say that $T\colon A\rightarrow \mathbb R$ is \emph{constant for $i\in N$} if for all $a_{-i}\in A_{-i}$, $T(\cdot,a_{-i})$ is constant.
Then, a game $G$ is \emph{constant} if for all $i\in N$, $G_i$ is constant for $i$.

\begin{lemma}[Invariance of equilibria]\label{lem:invariance}
	Let $G$ be a game on $A$.
	Then, 
	\begin{enumerate}[leftmargin=*,label=\textit{(\roman*)}]
		\item for all $\alpha\in\mathbb R_{++}^N$ and $\beta\in\mathbb R^N$, $\nash(G) = \nash(\alpha G + \beta)$, and\label{item:invariance1}
		\item for all constant games $\tilde G$ on $A$, $\nash(G) = \nash(G + \tilde G)$.
		\label{item:invariance2}
	\end{enumerate}
\end{lemma}
\begin{proof}
	The statement follows from a straightforward calculation.
\end{proof}
For $n = 2$, $G$ is constant if every column for $G_1$ and every row of $G_2$ is constant.
Hence, \Cref{lem:invariance}\ref{item:invariance2} asserts that adding a constant to some column of $G_1$ or row of $G_2$ does not change the set of equilibria. 

The second type of modification of payoff functions concerns multiplication by tensors.
The following notation will be convenient.  

\begin{definition}[Hadamard product]\label{def:hadamard}
	Let $A\in\mathcal F(U)^n$ and $T,T'\colon A\rightarrow\mathbb R$.
	Then, the Hadamard product $T\hada T'\colon  A\rightarrow\mathbb R$ of $T$ and $T'$ is defined by
	\begin{align*}
		(T\hada T')(a) = T(a)T'(a)
	\end{align*}
	for all $a\in A$.
\end{definition}

\begin{lemma}[Contortion of equilibria]\label{lem:contortion}
	Let $G$ be a game on $A$.
	Let $i\in N$, $q\in\mathbb R_{++}^{A_i}$, and $T_q\colon A\rightarrow\mathbb R$ so that for all $a\in A$, $T_q(a) = q(a_i)$.
	Then,
	\begin{align*}
		p\in\nash(G) \quad\text{ if and only if }\quad (\tilde p_i,p_{-i})\in\nash(G^{i,T_q}), 
	\end{align*}
	where $G^{i,T_q}_i = G_i$, for all $j\neq i$, $G^{i,T_q}_j = G_j\hada T_q$, and for all $a_i\in A_i$, $\tilde p_i(a_i) = (p_i(a_i)/q(a_i))/(\sum_{a_i'\in A_i}p_i(a_i')/q(a_i'))$.
\end{lemma}

\begin{proof}
	The statement follows from a straightforward calculation.
\end{proof}

Note that $T_q$ in \Cref{lem:contortion} is constant for all $j\neq i$.
For $n = 2$, the game $G^{1,T_q}$ is obtained from $G$ by multiplying the $a_1$ \emph{row} of the matrix $G_2$ by $q(a_1)$ for all $a_1\in A_1$.

The next lemma show that in some sense, it is enough to consider games where the payoff function of every player is slice-stochastic. 
This is explained in more detail after the lemma.

\begin{lemma}[Universality of slice-stochastic tensors]\label{lem:slicestochasttransform}
	Let $A\in\mathcal F(U)^n$ with $|A_1| = \dots = |A_n|$.
	Let $p$ be a profile on $A$ so that all $p_i$ have full support.
	Then, there are $T_1,\dots,T_n\colon A\rightarrow\mathbb R_{++}$ so that $T_i$ is constant for all $j\neq i$ and for all games $G$ on $A$ and $p\in\nash(G)$, there is a game $\bar G$ on $A$ so that the following hold.
	\begin{enumerate}[leftmargin=*,label=\textit{(\roman*)}]
		\item For all $i\in N$, $\bar G_i = G_i\hada T^{-i} + S_i$, where $T^{-i}$ is the Hadamard product of all $T_j$ with $j\neq i$, and $S_i\colon A\rightarrow\mathbb R$ is constant for $i$. 
		\label{item:slicestochasttransform1}
		\item For all $i\in N$, $\bar G_i$ is slice-stochastic for $i$.\label{item:slicestochasttransform2}
		\item $\bar p\in\nash(\bar G)$, where $\bar p$ and $\bar p_i$ is the uniform distribution on $A_i$.\label{item:slicestochasttransform3}
	\end{enumerate}
\end{lemma}
\begin{proof}
	 For all $i\in N$, let $T_i\colon A\rightarrow\mathbb R_{++}$ so that for all $a\in A$, $T_i(a) = \eps p_i(a_i)$, where $\eps > 0$ is small compared to (the reciprocal of) the largest payoff in $G$ and $|A| = m^n$.
	Define a game $\hat G$ on $A$ so that for all $i\in N$, $\hat G_i = G_i\hada T^{-i}$. 
	It follows from \Cref{lem:contortion} that $\bar p\in\nash(\hat G)$.

	For all $i\in N$, let $S_i\colon A\rightarrow \mathbb R$ so that for all $a\in A$, $$S_i(a) = \frac1m\left(1 - \sum_{a_i'\in A_i} \hat G_i(a_i',a_{-i})\right).$$
	Note that $S_i$ is constant for $i$.
	Let $\bar G$ be the game so that for all $i\in N$, $\bar G_i = \hat G_i + S_i$.
	By \Cref{lem:invariance}\ref{item:invariance2}, $\bar p\in\nash(\bar G)$ and so for all $i\in N$, $\sum_{a_{-i}\in A_{-i}} \bar G_i(\cdot,a_{-i})$
	is constant on $A_i$.
	Moreover, for all $a_{-i}\in A_{-i}$,
	$\sum_{a_i\in A_i} \bar G_i(a_i,a_{-i}) = 1$
	by the choice of $S_i$.
	Lastly, since $\eps$ is small, $\hat G_i(a)\approx 0$ and $S_i(a)\approx\frac1m$ for all $a\in A$.
	Hence, $\bar G_i$ is slice-stochastic.
	We have thus shown that $\bar G$ satisfies all three conditions in the statement of the lemma.
\end{proof}

The condition \ref{item:slicestochasttransform3} in \Cref{lem:slicestochasttransform} is redundant since the strategy profile where every player distributes uniformly automatically an equilibrium if the payoff function of every player is slice-stochastic (by \Cref{item:slice2} in the definition of slice-stochasticity).
We chose to state it to make it explicit.

The way in which we will use \Cref{lem:slicestochasttransform} is as follows.
Given a game $G$ with the same number of actions for every player and a full support equilibrium $p\in\nash(G)$, we want to show that $p\in f(G)$ if $f$ is a nice total solution concept. 
To do this, we show that the game $\bar G$ obtained from the lemma (by virtue of having slice-stochastic tensors as payoff functions) can essentially be written as a convex combination of games for which $\bar p$ is the unique equilibrium.
Applying the inverse transformation of that in \ref{item:slicestochasttransform1} to each of the summands and using \Cref{lem:contortion} shows that $G$ can essentially be written as a convex combination of games for which $p$ is the unique equilibrium. 
Using consistency and the fact that $f\subseteq\nash$, this shows that $p\in f(G)$.
The caveat ``essentially'' refers to the fact that one may need to multiply the $G_i$ by positive scalars, add constant games to $G$, and clone actions to get the desired representation as a convex combination.
Since neither of the first two operations changes the set of equilibria and the effect of introducing clones is controlled by consequentialism, this does not cause problems.
Similarly, the restriction that every player has the same number of actions in $G$ is without loss of generality by consequentialism. 

Together, \Cref{lem:playeripayoff}, \Cref{lem:bvnslice}, and \Cref{lem:slicestochasttransform} show the following.
If we want to show that $p\in f(G)$ whenever $p\in\nash(G)$ and $p$ has full support, then it is enough to do so in the case when $G$ is a player $i$ payoff game with $G_i$ deterministic slice-stochastic.
The full support assumption will be successively eliminated via Lemmas~\ref{lem:conversefullsupportrat}, \ref{lem:quasi-strict}, and~\ref{lem:quasistricttoall}.

\subsection{Decomposition of Deterministic Slice-Stochastic Tensors}  

The first step is to construct a sufficiently rich class of games that have the strategy profile where every player randomizes uniformly as their unique equilibrium.
This will be the class of cyclic games and almost cyclic games.
In cyclic games, the payoff of every player only depends on the action of one other player and the dependencies form a cycle.
Roughly, a player gets a payoff of 1 if she matches the action of the preceding player and 0 otherwise.
The fact that every such game has uniform randomization as an equilibrium is not hard to see. 
Uniqueness is achieved by imposing a restriction on the notion of ``matching'' (see \Cref{def:cyclic}\ref{item:cyclic1}).
Almost cyclic games differ only insofar that there is one exceptional player whose payoff not only depends on the action of the preceding player in the cycle but (for few action profiles) also on the actions of all other players.
Making the concepts above precise requires two definitions.

\begin{definition}[Fixed subsets]\label{def:fixedsubsets}
	Let $A$ be a set and $\pi\in\Sigma_A$.
	Then, $B\subseteq A$ is a \emph{fixed subset} of $\pi$ if $\pi(B) = B$.
	We say that $\pi$ has \emph{no non-trivial fixed subset} if its only fixed subsets are $\emptyset$ and $A$.
\end{definition}

For any two permutations $\pi,\pi'$, $\pi'\circ\pi$ has a non-trivial fixed subset if and only if $\pi\circ\pi'$ does.
At least when $A$ is finite, permutations without a non-trivial fixed subset always exist. 
Any cyclic permutation is an example.
Also, for any permutation $\pi$, there is a permutation $\pi'$ so that $\pi'\circ \pi$ has no non-trivial fixed subset.

\begin{definition}[Permutation sets]\label{def:permutationsets}
	Let $A_1,\dots,A_n\in\mathcal F(U)$ with $|A_1| = \dots = |A_n|$ and let $A = A_1\timesdots A_n$.
	A set $A^*\subseteq A$ is a \emph{permutation set} if for all $i\in N$ and $a_i'\in A_i$, there is exactly one $a\in A^*$ with $a_i = a_i'$.
\end{definition}

Another way of saying that $A^*$ is a permutation set is that the projection from $A^*$ onto each $A_i$ is bijective.
If $A_1 = \dots = A_n = [m]$, yet another way is requiring that there are permutations $\pi_1,\dots,\pi_{n-1}$ of $[m]$ so that $A^* = \{(k,\pi_1(k),\dots,\pi_{n-1}\circ\dots\circ\pi_1(k))\colon k\in[m]\}$.
This characterization motivates the terminology.

\begin{definition}[(Almost) cyclic games]\label{def:cyclic}
	Let $A = [m]^n$ and $G$ be a game on $A$.
	We say that $G$ is \emph{cyclic} if there are $\pi_1,\dots,\pi_n\in\Sigma_{[m]}$ and $\alpha\in\mathbb R_{++}^N$ so that the following holds.
	\begin{enumerate}[leftmargin=*,label=\textit{(\roman*)}]
		\item $\pi_n\circ\dots\circ\pi_1$ has no non-trivial fixed subset,\label{item:cyclic1}
		\item for all $j\in N$ and $a\in A$,
		\begin{align*}
			G_j(a) = 
			\begin{cases}
				\alpha_j \quad&\text{if } a_j = \pi_{j-1}(a_{j-1}), \\ 
				0 &\text{otherwise.}
			\end{cases}
		\end{align*}\label{item:cyclic2}
	\end{enumerate}
	For $n \ge 3$, we say that $G$ is \emph{almost cyclic} if there are $i\in N$, $\pi_1,\dots,\pi_n\in\Sigma_{[m]}$, $A^*\subseteq A$, and $\alpha\in\mathbb R_{++}^N$ so that the following holds.
	\begin{enumerate}[leftmargin=*,label=\textit{(\roman*)}]
		\item $\pi_n\circ\dots\circ\pi_1$ has no non-trivial fixed subset,\label{item:almostcyclic1}
		\item for all $j\neq i$ and $a\in A$,
		\begin{align*}
			G_j(a) &= 
			\begin{cases}
				\alpha_j \quad&\text{if } a_j = \pi_{j-1}(a_{j-1}), \\ 
				0 &\text{otherwise, }
			\end{cases}
			\text{ and }\\
			G_i(a) &= 
			\begin{cases}
				\alpha_i \quad&\text{if } a_i = \pi_{i-1}(a_{i-1}) \text{ and }  a\not\in A^*\\
				0 &\text{otherwise, }
			\end{cases}
		\end{align*}\label{item:almostcyclic2}
		\item $A^*$ is a permutation set and for all $a\in A^*$, $a_i = \pi_{i-1}(a_{i-1})$ and there is $j\neq i,i+1$ with $a_j \neq \pi_{j-1}(a_{j-1})$.\label{item:almostcyclic3}
	\end{enumerate}
\end{definition}

Unless otherwise noted, we assume that $\alpha = (1,\dots,1)$.

\begin{example}[(Almost) cyclic games for $n = 2,3$]\label{ex:cyclicgame}
	Let $n = 2$ and consider the following game $G$ where both players have three actions.
	The first (second) entry in each cell denotes the payoff of player 1 (player 2).
	\begin{center}
	\begin{tikzpicture}[baseline,
  label distance=10pt]
		\matrix [matrix of math nodes, left delimiter=(,right delimiter=), row sep=.1cm, column sep=.1cm] (g) {
		      1,0 & 0,1 & 0,0\\ 
				0,0 & 1,0 & 0,1\\ 
				0,1 & 0,0 & 1,0\\ };

	\end{tikzpicture}
	\end{center}
	Then, $G$ is a cyclic game with $\pi_1 = (123)$ and $\pi_2 = (1)(2)(3)$.
	Almost cyclic games for two players are degenerate in the sense that the payoff function of the exceptional player is $0$.
	
	Now let $m = n = 3$, $i = 1$, $\pi_1 = (123)$, and $\pi_2 = \pi_3 = (1)(2)(3)$.
	Clearly, $\pi_3\circ\pi_2\circ\pi_1$ has no non-trivial fixed subset.
	One can check that the permutation set $A^* = \{(1,2,1),(2,3,2),(3,1,3)\}$ satisfies \ref{item:almostcyclic3}.
	The payoff function of player 1 is shown below (player 1 chooses the matrix, player two the row, and player 3 the column; the entries corresponding to $A^*$ are marked in boldface).
  	\begin{center}
  	\begin{tikzpicture}[baseline,
    label distance=10pt]
  		\matrix [matrix of math nodes, left delimiter=(,right delimiter=), row sep=.1cm, column sep=.1cm] (g1) {
  		      1 & 0 & 0\\ 
  				\textbf{0} & 0 & 0\\ 
  				1 & 0 & 0\\ };
				
  		\matrix [right = 2cm of g1, matrix of math nodes, left delimiter=(,right delimiter=), row sep=.1cm, column sep=.1cm] (g2) {
  		      0 & 1 & 0\\ 
  				0 & 1 & 0\\ 
  				0 & \textbf{0} & 0\\ };
				
  		\matrix [right = 2cm of g2, matrix of math nodes, left delimiter=(,right delimiter=), row sep=.1cm, column sep=.1cm] (g3) {
  		      0 & 0 & \textbf{0}\\ 
  				0 & 0 & 1\\ 
  				0 & 0 & 1\\ };

  	\end{tikzpicture}
  	\end{center}
\end{example}

\begin{lemma}[Equilibria of (almost) cyclic games]\label{lem:cyclicgamesequilibria}
	Let $G$ be a cyclic game or an almost cyclic game.
	Then, $G$ has a unique Nash equilibrium where every player randomizes uniformly over all of her actions.
\end{lemma}
\begin{proof}
	We prove the statement for almost cyclic games. 
	The proof for cyclic games is easier.
	More specifically, for cyclic games, \Cref{case:almostcyclic1} and \Cref{case:almostcyclic2} below can be combined into one that is proved in the same way as \Cref{case:almostcyclic1}.
	
	Let $n \ge 3$.
	Let $i\in N$, $\pi_1,\dots,\pi_n\in\Sigma_{[m]}$, $A^*$, and $\alpha$ be as in the definition of almost cyclic games.
	By \Cref{lem:invariance}, it is without loss of generality to assume that $\alpha = (1,\dots,1)$.
	Let $p$ be the strategy profile where $p_j$ is the uniform distribution on $A_j$ for all $j\in N$.
	
	It is easy to see that $p$ is an equilibrium since for all $a_i\in A_i$,
	\begin{align*}
		\sum_{a_{-i}\in A_{-i}} G_i(a_i,a_{-i}) = m^{n-2} - 1,
	\end{align*}
	where the $-1$ comes from the fact that $A^*$ is a permutation set.
	Moreover, for all $j\neq i$ and $a_j\in A_j$, 
	\begin{align*}
		\sum_{a_{-j}\in A_{-j}} G_j(a_j,a_{-j}) = m^{n-2}.
	\end{align*}
	
	Now we show that $p$ is the unique equilibrium.
	Assume that $p'$ is an equilibrium.
	For all $j\in N$, let $B_j = \arg\max_{a_j\in A_j} p_j'(a_j)$.
	Note that for all $j\neq i$, the payoff of $j$ only depends on her own strategy and the strategy of $j-1$.
	Hence, $a_j\in A_j$ is a best response to $p'_{-j}$ if $a_j \in \pi_{j-1}(B_{j-1})$.
	So $B_j\subseteq\supp(p_j')\subseteq \pi_{j-1}(B_{j-1})$.
	In particular, 
	\begin{align}
		|\supp(p_i')| \ge |\supp(p_{i+1}')| \ge \dots \ge |\supp(p_{i-1}')|.\label{eq:supportineq}
	\end{align}
	We distinguish two cases.
	\begin{case}
		Suppose that $|\supp(p_j')| \ge 2$ for some $j\neq i$.
		By \eqref{eq:supportineq}, $|\supp(p'_{i+1})| \ge 2$.
		Since $n\ge 3$ and $A^*$ is a permutation set, it follows that $G_i(a_i,p'_{-i}) > 0$ for all $a_i\in \pi_{i-1}(\supp(p_{i-1}'))$.
		Hence, $\supp(p_i')\subseteq \pi_{i-1}(\supp(p_{i-1}'))$, and so $|\supp(p_i')| \le |\supp(p_{i-1}')|$.
		It follows that all inequalities in \eqref{eq:supportineq} hold with equality.
		But then for all $j\in N$, $\supp(p'_j) = \pi_{j-1}(\supp(p'_{j-1}))$, and so $\supp(p'_1)$ is a fixed subset of $\pi_n\circ\dots\circ\pi_1$.
		By assumption, this is only possible if $A_1 = \supp(p'_1)$, which in turn implies that $\supp(p'_j) = A_j$ for all $j\in N$.
		Hence, $p' = p$.\label{case:almostcyclic1}
	\end{case}
	\begin{case}
		Suppose that $|\supp(p_j')| = 1$ for all $j\neq i$ and let $a'_j\in A_j$ so that $p'_j(a'_j) = 1$.
		Note that $a'_j = \pi_{j-1}(a'_{j-1})$ for all $j\neq i,i+1$.
		Moreover, we have that $G_i(\pi_{i-1}(a'_{i-1}),a'_{-i}) = 1$ unless $(\pi_{i-1}(a'_{i-1}),a'_{-i}) \in A^*$.
		But by the second part of \ref{item:almostcyclic3} in the definition of almost cyclic games, $(\pi_{i-1}(a'_{i-1}),a'_{-i}) \not\in A^*$.
		Hence, $G_i(\pi_{i-1}(a'_{i-1}),a'_{-i}) = 1$.
		By \ref{item:almostcyclic2} and the fact that $|\supp(p'_{i-1})| = 1$, $G_i(a_i,p'_{-i}) = 0$ unless $a_i = a_i' = \pi_{i-1}(a'_{i-1})$.
		Thus, $a'_i$ is the unique best response of player $i$ to $p'_{-i}$.
		Since $p'$ is an equilibrium, it follows that $p'_i(a'_i) = 1$ and $|\supp(p'_i)| = 1$.
		Finally, since $p'_{i+1}$ is a best response to $p'_i$, we have that $a'_{i+1} = \pi_i(a'_i)$.
		So $a'_j = \pi_{j-1}(a'_{j-1})$ for all $j\in N$, which means that $\{a'_1\}$ is a fixed subset of $\pi_n\circ\dots\circ\pi_1$ and contradicts \ref{item:almostcyclic1}.\label{case:almostcyclic2}
	\end{case}
\end{proof}

\begin{definition}[Permutation tensors and permutation games]
	Let $A_1,\dots,A_n\in\mathcal F(U)$ with $|A_1| = \dots = |A_n|$ and let $A = A_1\timesdots A_n$.
	A function $T\colon A\rightarrow\mathbb R$ is a \emph{permutation tensor} if there is a permutation set $A^*\subseteq A$ so that $T(a) = 1$ for all $a\in A^*$ and $T(a) = 0$ for all $a\in A\setm A^*$.
	A game $G$ is a permutation game if there is $i\in N$ so that $G$ is a player $i$ payoff game and $G_i$ is a permutation tensor.
\end{definition}

For $n = 2$, a permutation tensor is a permutation matrix.

We show that every permutation game can be written as a convex combination of cyclic games and almost cyclic games up to an additive constant.
Note that even though permutation games are player $i$ payoff games, the games in the decomposition are not.

\begin{lemma}[Decomposition of permutation games]\label{lem:permutationgamesdecomp}
	Let $A_1,\dots,A_n\in\mathcal F(U)$ so that $|A_1| = \dots = |A_n| = m$, and let $A = A_1\timesdots A_n$.
	Let $i\in N$ and $A^*\subseteq A$ be a permutation set.
	Let $G$ be the permutation game for $i$ and $A^*$ on $A$.
	Then, $G$ can be written as a convex combination of cyclic games, almost cyclic games, and a constant $\beta\in\mathbb R^N$. 
\end{lemma}
\begin{proof}
	Let $G$ be a game as in the statement of the lemma.
	For simplicity, assume that $i = n$ and for all $j\in N$, $A_j = [m]$.

	Let $\pi_1,\dots,\pi_n\in\Sigma_{[m]}$ so that $A^* = \{(k,\pi_1(k),\dots,\pi_{n-1}\circ\dots\circ\pi_1(k))\in[m]^N\colon k\in[m]\}$ and $\pi_n\circ\dots\circ\pi_1$ has no non-trivial fixed subset.
	This is possible since we may first choose $\pi_1,\dots\pi_{n-1}$ so that the first condition holds (using that $A^*$ is a permutation set) and then choose $\pi_n$ so that $\pi_n\circ\dots\circ\pi_1$ has no non-trivial fixed subset.

	Let $\hat A = \{a\in A\colon a_n = \pi_{n-1}(a_{n-1})\}$.
	Note that $|\hat A| = m^{n-1}$ and $A^*\subseteq\hat A$.
	Let $B^1,\dots,B^M\subseteq A$ be a partition of $\hat A\setm A^*$ into permutation sets, where $M = m^{n-2} - 1$.
	(Note that $M = 0$ if $n = 2$.)	
	For example, one may take $$B^{s_{-\{1,n\}}} = \{(a_1,a_2 + s_2,\dots,a_{n-1} + s_{n-1}, \pi_{n-1}(a_{n-1} + s_{n-1}))\in[m]^N\colon a\in A^*\},$$ where $s_{-\{1,n\}}\in [m]^{N\setm\{1,n\}}\setm\{0\}$.
	Since $s_{-\{1,n\}}\neq 0$, the last condition in \ref{item:almostcyclic3} holds for $B^l$.
	For all $l\in[M]$, let $G^l$ be the almost cyclic game for $i = n$, $\pi_1,\dots,\pi_n$, $B^l$, and $\alpha = (1,\dots,1)$.
	In particular, $G^l_n(a) = 0$ for all $a\in B^l$ and $G_n^l(a) = 1$ for all $a\in \hat A\setm B^l$.
	
	Now let $$\hat G = \sum_{l\in[M]} G^l.$$
	Then, the following hold.
	\begin{itemize}
		\item For all $a\in A^*$, $G_n(a) = M$; for all $a\in \hat A\setm A^*$, $G_n(a) = M-1$.
		\item For all $j\neq n$ and $a\in A$ with $a_j = \pi_{j-1}(a_{j-1})$, $G_j(a) = M$.
		\item For all $j\in N$ and $a\in A$ with $a_j \neq \pi_{j-1}(a_{j-1})$, $G_j(a) = 0$.
	\end{itemize}
	
	Recall that for all $j\in N$ and $(a_{j-1},a_j)\in [m]^2$, there is a cyclic game $G'$ for some $\pi_1',\dots,\pi_n'\in\Sigma_{[m]}$ with $a_j = \pi'_{j-1}(a_{j-1})$ and arbitrary $\alpha'\in\mathbb R_{++}^N$ (see the remarks after \Cref{def:fixedsubsets}).
	For all $j\in N$ and $(a_{j-1},a_j)\in [m]^2$ with $a_j \neq \pi_{j-1}(a_{j-1})$, let $G^{j,a_{j-1},a_j}$ be a cyclic game for some $\pi_1',\dots,\pi_n'\in\Pi_{[m]}$ with $a_j = \pi'_{j-1}(a_{j-1})$, $\alpha'_{j'} = 1$ for all $j'\neq j$, and $\alpha'_j = M+1$ if $j\neq n$ and $\alpha'_j = M$ if $j = n$.
	Then, let
	\begin{align}
		\tilde G = \hat G + \sum G^{j,a_{j-1},a_j} + \sum G',\label{eq:Gtilde}
	\end{align}
	where the first sum ranges over all $j\in N$ and $(a_{j-1},a_j)\in[m]^2$ with $a_j \neq \pi_{j-1}(a_{j-1})$, and the second sum ranges over all cyclic games with $\alpha = (1,\dots,1)$ whose tuple of permutations does not already appear in one of the games in the first sum.
	Since for all $j\in N$ and $(a_{j-1},a_j)\in[m]^2$, the number of cyclic games with $a_j = \pi'_{j-1}(a_{j-1})$ is the same (assuming $\alpha$ is fixed), the following hold.
	\begin{itemize}
		\item For all $a^*\in A^*$ and $a\in A\setm A^*$, $\tilde G_n(a^*) = \tilde G_n(a) + 1$.
		\item For all $j\neq n$ and $a,a'\in A$, $G_j(a) = G_j(a')$. 
	\end{itemize}
	Hence, $G = \tilde G + \beta$ for some $\beta\in\mathbb R^N$.
	More explicitly, if $L$ is the total number of games in the first and second sum in \eqref{eq:Gtilde}, then
	\begin{itemize}
		\item For all $a\in A^*$, $\tilde G_n(a) = M + \frac Lm$; for all $a\in A\setm\bar A^*$, $G_n(a) = M - 1 + \frac Lm$.
		\item For all $j\neq i$ and $a\in A$, $\tilde G_j(a) = M + \frac Lm$.
	\end{itemize}
	(The denominator $m$ comes from the fact that for a permutation game $G'$, the fraction of actions for which $G'_j(a) = \alpha_j$ is $\frac1m$.)
	This proves $G$ can be written as a sum of games of the claimed types.
	Multiplying each game in that sum by the same appropriately chosen positive scalar gives the representation of $G$ as a convex combination.
\end{proof}

The last lemma in this section roughly shows that every game with deterministic slice-stochastic tensors as payoff functions can be written as a convex combination of permutation games.
This conclusion is not literally true. 
More precisely, we show that every game with deterministic slice-stochastic tensors as payoff functions is the blow-down of a convex combination of permutation games (with the same number of clones of every action for every player).
By \Cref{lem:cyclicgamesequilibria} and \Cref{lem:permutationgamesdecomp}, each of the games in this sum can, in turn, be written as a convex combination of games for which uniform randomization is the unique equilibrium. 
Consistency and consequentialism thus imply that uniform randomization has to be returned in the original game, possibly alongside other equilibria.

\begin{lemma}[Decomposition of slice-stochastic games]\label{lem:slicestochasticdecomposition}
	Let $A_1,\dots,A_n\in\mathcal F(U)$ with $|A_1| = \dots = |A_n| = m$, and let $A = A_1\timesdots A_n$.
	Let $G$ be a game on $A$ so that $G_i$ is a deterministic slice-stochastic tensor for all $i\in N$.
	Let $p$ be the profile on $A$ so that each $p_i$ is the uniform distribution on $A_i$.
	Then, up to cloning actions, multiplying by positive scalars, and adding constants, $G$ can be written as a convex combination of games for which $p$ is the unique equilibrium.
	More precisely, there are $\alpha\in\mathbb R_{++}^N$, $\beta\in\mathbb R^N$, and a game $\bar G$ so that
	\begin{enumerate}[leftmargin=*,label=\textit{(\roman*)}]
		\item $G$ is a blow-down of $\alpha\bar G + \beta$ with surjection $\phi$,
		\item $\phi_*(\hat p) = p$, where $\hat p_j$ is the uniform distribution on the actions of player $j$ in $\bar G$, and
		\item $\bar G$ is a convex combination of games for which $\hat p$ is the unique equilibrium.
	\end{enumerate}
\end{lemma}
\begin{proof}
	First, we observe that it suffices to consider the case that $G$ is a player $i$ payoff game (meaning that $G_j = 0$ for all $j\neq i$).
	The idea now is to ``blow up'' $G$ by introducing $m^{n-2}$ actions for each action of every player in $G$.
	Then, one can define a permutation game $\hat G$ on the larger action sets so that for every action profile $a$ in $G$ for which player $i$ has payoff 1, there is exactly one blow up $\hat a$ of $a$ for which $i$ has payoff $1$ in $\hat G$.
	Since $\hat G$ is a permutation game, we know from \Cref{lem:permutationgamesdecomp} that, up to adding constants, it can be written as a convex combination of games for which the strategy profile where every player plays the uniform distribution is the unique equilibrium.
	Now permuting all actions in $\hat G$ that come from blowing up the same action in $G$, summing up over all the resulting games, and multiplying by a positive scalar gives a game $\bar G$ that is a blow-up of $G$.
	This achieves the desired decomposition.
	(\Cref{rem:permutationgamedecomp} explains why the blowing up is necessary.)

	As noted above, we may assume that there is $i\in N$ so that $G$ is a player $i$ payoff game.
	Let $A^* = \{a\in A\colon G_i(a) = 1\}$ be the actions for which player $i$ has payoff 1 in $G$.
	For all $j\in N$ and $a_j\in [m]$, let $B_j^{a_j} = \{a_{-j}\in A_{-j}\colon (a_j,a_{-j}) \in A^*\}$ be the set of opponents action profiles for which $i$ gets payoff $1$ when $j$ plays $a_j$.
	Since $G_i$ is slice-stochastic, $|B_j^{a_j}| = m^{n-2}$. 
	(For $j = i$, this is \ref{item:slice2} in the definition of slice-stochastic games; for $j\neq i$, for all $a_j$ and $a_{-\{i,j\}}$, there is exactly one $a_i$ so that $G_i(a_i,a_j,a_{-\{i,j\}}) = 1$ by \ref{item:slice1}.)
	Let $\hat A_j = \bigcup_{a_j\in A_j} \{a_j\}\times B_j^{a_j}$ and $\hat A = \hat A_1\timesdots \hat A_n$.
	Note that $\hat A_j$ has size $m^{n-1}$.
	Now let  and
	\begin{align*}
		\hat A^* = \{((a_1,a_{-1}),\dots,(a_n,a_{-n}))\in\hat A\colon a\in A^*\}.
	\end{align*}
	It is easy to see that $\hat A^*$ is a permutation set in $\hat A$.
	Let $\hat G$ be the permutation game for $i$ on $\hat A$ for the permutation set $\hat A^*$.
	By \Cref{lem:permutationgamesdecomp}, we have that $\hat G$ can be written as a sum of cyclic games, almost cyclic games, and a constant.
	Moreover, by \Cref{lem:cyclicgamesequilibria}, we know that $\hat p$ is the unique equilibrium of cyclic and almost cyclic games, where $\hat p_j$ is the uniform distribution on $\hat A_j$.
	
	For all $j\in N$, let $\phi_j\colon \hat A_j\rightarrow A_j$ be projection onto the first coordinate, and let $\phi = (\phi_1,\dots,\phi_n)$.
	That is, $\phi_j((a_j,a_{-j})) = a_j$.
	Let $\Sigma_j\subseteq\Sigma_{\hat A_j}$ be the set of all permutations $\pi_j$ so that $\phi_j(\hat a_j) = \phi_j(\pi_j(\hat a_j))$ for all $\hat a_j\in\hat A_j$ (that is, all permutations that keep the first coordinate fixed).
	Let $\Sigma = \Sigma_1\timesdots \Sigma_n$.
	So we have that for all $\hat a\in\hat A$ and $\pi\in\Sigma$,
	\begin{align*}
		G_i(\phi(\hat a)) = G_i((\phi\circ\pi)(\hat a)).
	\end{align*}
	Consider the game
	\begin{align*}
		\bar G = \frac M{|\Sigma|} \sum_{\pi\in\Sigma} \hat G\circ\pi,
	\end{align*}
	where $M = m^{(n-1)n}$.
	(The factor $M$ is necessary since for all $a\in A$, there are $M$ profiles $\hat a\in \hat A$ with $\phi(\hat a) = a$.)

	Since each $\hat a\in\hat A^*$ is completely determined by its first coordinates, for all $a\in A^*$, there is exactly one $\hat a\in\hat A^*$ with $\phi(\hat a) = a$.
	Thus, the following hold.
	\begin{itemize}
		\item For all $\hat a\in\hat A$ with $\phi(\hat a)\in A^*$, $\bar G_i(\hat a) = 1$. 
		\item For all $\hat a\in\hat A$ with $\phi(\hat a)\in A\setm A^*$, $\bar G_i(\hat a) = 0$.
		\item For all $j\neq i$ and $\hat a\in\hat A$, $\bar G_j(\hat a) = 0$.
	\end{itemize}
	So for all $j\in N$ and $a_j\in A_j$, all actions in $\{a_j\}\times B_j^{a_j}$ are clones of each other in $\bar G$.
	This gives that $\bar G$ is a blow-up of $G$.
	Moreover, $\phi_*(\hat p) = p$ since for all $j\in N$ and $a_j\in A_j$, there is the same number of clones (namely $m^{n-2}$) in $\hat A_j$.
	This gives the desired decomposition of $G$.
\end{proof}

\begin{remark}[Decomposition into permutation games]\label{rem:permutationgamedecomp}
	The ``blowing up'' of $G$ to $\bar G$ in the proof of \Cref{lem:slicestochasticdecomposition} by introducing $m^{n-2}$ clones of every action is necessary since not every deterministic slice-stochastic tensor can be written as a sum of permutation tensors.
	For example, let $n = 3$, $m = 2$, $i = 1$, and $G_1(a) = 1$ for $a\in\{(1,1,1),(1,2,2),(2,1,2),(2,2,1)\}\subseteq\{1,2\}^3$ and $G_1(a) = 0$ otherwise.
	Then, there is no permutation set $B\subseteq\{1,2\}^3$ so that $G_1(a) = 1$ for all $a\in B$.  
\end{remark}

So far, we have shown the following.
To prove that $p\in f(G)$ whenever $p\in\nash(G)$ and $p$ has full support, it suffices to show this for the case when each $G_i$ is slice-stochastic by \Cref{lem:slicestochasttransform}. 
By the generalization of the Birkhoff-von Neumann theorem, \Cref{lem:bvnslice}, the fact that in every such game uniform randomization is an equilibrium, and consistency, we can further restrict to $G_i$'s that are deterministic slice-stochastic.
Now \Cref{lem:slicestochasticdecomposition} shows that all those games can, in essence, be written as convex combinations of games for which uniform randomization is the unique equilibrium and, thus, has to be returned by $f$.
Consistency then gives the desired conclusion.
The last step is to extend the statement beyond full support equilibria.

\subsection{Reduction to Full Support Equilibria}\label{sec:fullsupportreduction}

We show that it suffices to prove that all full support equilibria have to be returned, which we have done in the previous section.
There are three steps to the argument. 
First, we reduce to equilibria with support equal to the set of all rationalizable actions, then to quasi-strict equilibria, and lastly to arbitrary equilibria.\footnote{We say that an action (profile) is rationalizable if it survives iterated elimination of strictly dominated actions.}
The strategy is always to write a game with some type of equilibrium as a convex combination of games where the same equilibrium is of the type in the preceding step.
For example, \Cref{lem:conversefullsupportrat} shows that any nice total solution concept has to return all equilibria whose support consists of all rationalizable actions.

\begin{lemma}[Equilibria with full support on rationalizable actions]\label{lem:conversefullsupportrat}
	Let $f$ be a nice total solution concept.
	Let $G$ be a game on $A$.
	Let $\bar A_i\subseteq A_i$ be the sets of rationalizable actions of $i\in N$ in $G$ and $\bar A = \bar A_1\timesdots\bar A_n$.
	Then, if $p\in \nash(G)$ so that for all $i\in N$, $\supp(p_i) = \bar A_i$, then $p \in f(G)$.
\end{lemma}
\begin{proof}
	Consider a decreasing sequence of action profiles obtained by successively removing dominated actions until no more deletions are possible.
	That is, let $(A_1^0,\dots,A_n^0),\dots,(A_1^K,\dots,A_n^K)\in 2^{A_1}\timesdots 2^{A_n}$ so that for all $k\in[K]$, there is $i\in N$ for which the following holds.
	\begin{itemize}
		\item $A_i^k \subsetneq A_i^{k-1}$ and for all $j\neq i$, $A_j^k = A_j^{k-1}$.
		\item For all $a_i\in A_i^{k-1}\setm A_i^k$, there is an action $\psi(a_i) \in A_i^k$ that dominates $a_i$ when restricting $G$ to $A_1^{k-1}\timesdots A_n^{k-1}$.
		\item For all $j\in N$, $A_j^0 = A_j$ and $A_j^K = \bar A_j$.
	\end{itemize}
	For all $i\in N$ and $a_i\in A_i\setm \bar A_i$, let $\bar\psi(a_i) = \psi^s(a_i)\in\bar A_i$, where $s\in\mathbb N$ is the unique power so that $\psi^s(a_i)\in \bar A_i$ (here we mean $\psi$ applied $s$ times).
	
	Denote by $\bar G$ the game $G$ restricted to action profiles in $\bar A$.
	We make several reductions.
	By consequentialism, we may assume that $|\bar A_1| = \dots = |\bar A_n|$.
	By \Cref{lem:contortion} and \Cref{lem:slicestochasttransform}, we may assume that for all $i\in N$, $\bar G_i$ is slice-stochastic and $p_i$ is the uniform distribution on $\bar A_i$.
	By consequentialism, \Cref{lem:slicestochasticdecomposition}, and \Cref{lem:invariance}, we may further assume that $\bar G$ can be written as a convex combination of games for which $p$ is the unique equilibrium.
	That is, there are games $\bar G^1,\dots,\bar G^M$ on $\bar A$ so that for all $m\in[M]$, $p$ is the unique equilibrium of $\bar G^{m}$, and
	\begin{align}
		\bar G = \frac1M \sum_{m\in[M]} \bar G^{m}.\label{eq:ghatconvex}
	\end{align} 
	
	For all $m\in[M]$, we define a game $G^{m}$ on $A$ so that for all $i\in N$ and $a\in A$,
	\begin{align*}
		G_i^{m}(a) =
		\begin{cases}
			\bar G_i^{m}(a)\quad&\text{if } a\in \bar A,\\
			G_i(a) + \bar G_i^{m}(\bar\psi(a_i),a_{-i}) - G_i(\bar\psi(a_i),a_{-i}) &\text{if } a \in (A_i\setm\bar A_i)\times \bar A_{-i}\text{, and}\\
			G_i(a)&\text{if } a_{-i}\in A_{-i}\setm\bar A_{-i}.
		\end{cases}
	\end{align*}
	Since $\bar\psi(a_i) = \bar\psi(a_i')$ for $a_i,a_i'\in A_i\setm \bar A_i$ with $\psi(a_i) = a_i'$, we have that $\bar A$ is the set of rationalizable action profiles of $G^{m}$.
	It follows that $p$ is the unique equilibrium in $G^{m}$ and so $p\in f(G^m)$.
	Observe that for $a \in (A_i\setm\bar A_i)\times \bar A_{-i}$, we have by \eqref{eq:ghatconvex} that
	\begin{align*}
		\sum_{m\in[M]} \bar G_i^{m}(\bar\psi(a_i),a_{-i}) - G_i(\bar\psi(a_i),a_{-i}) = 0.
	\end{align*}
	Hence, $G = \frac1m\sum_{m\in[M]} G^{m}$.
	Consistency then implies that $p\in f(G)$.
\end{proof}

A profile $p$ is a quasi-strict equilibrium of $G$ if $p\in\nash(G)$ and
\begin{align*}
	G_i(a_i,p_{-i}) > G_i(a_i',p_{-i}) \text{ for all $a_i\in\supp(p_i),\, a_i'\in A_i\setminus\supp(p_i)$, and $i\in N$.}
\end{align*}
We show that nice total solution concepts have to return quasi-strict equilibria.

\begin{lemma}[Quasi-strict equilibria]\label{lem:quasi-strict}
 	Let $f$ be a nice total solution concept.
 	Let $G$ be a game on $A$.
	Then, if $p\in \nash(G)$ is quasi-strict, then $p \in f(G)$.
\end{lemma}
\begin{proof}
	For all $i\in N$, let $\hat A_i = \supp(p_i)$ and $\hat A = \hat A_1\timesdots \hat A_n$.
	By consequentialism, we may assume that $|\hat A_i| = 2|A_i\setm \hat A_i|$ and the number of clones of each action in $\hat A_i$ is even.
	Moreover, by \Cref{lem:slicestochasttransform} and the remarks thereafter, we may assume that $p_i$ is the uniform distribution on $\hat A_i$.
	Write $\hat A_i = \{a_i^1,\dots,a_i^K,b_i^1,\dots,b_i^K\}$ so that $a_i^{k}$ and $b_i^{k}$ are clones for all $k\in[K]$ and $A_i\setm \hat A_i = \{c_i^1,\dots,c_i^K\}$.
	The idea is to write $G$ as a convex combination of two games $G^1,G^2$ for which all actions in $A_i\setm \hat A_i$ are dominated and $p$ is an equilibrium of $G^1,G^2$. 
	\Cref{lem:conversefullsupportrat} and consistency of $f$ will then give that $p\in f(G)$.
	
	Let $i\in N$.
	Since $p$ is quasi-strict and $p_{-i}$ is uniform on $\hat A_{-i}$, we have for all $a_i\in \hat A_i$ and $a_i'\in A_i\setm \hat A_i$,
	\begin{align}
		\sum_{a_{-i}\in\hat A_{-i}} G_i(a_i,a_{-i}) > \sum_{a_{-i}\in\hat A_{-i}} G_i(a_i',a_{-i}).\label{eq:sumineq}
	\end{align}
	For all $k\in[K]$, let $v_i^{k}\in\mathbb R^{A_{-i}}$ so that
	\begin{align}
		\sum_{a_{-i}\in \hat A_{-i}} v_i^{k}(a_{-i}) = 0\label{eq:vsum0}
	\end{align}
	and for all $a_{-i}\in A_{-i}$,
	\begin{align}
		G_i(a_i^{k},a_{-i}) + v_i^{k}(a_{-i}) = G_i(b_i^{k},a_{-i}) + v_i^{k}(a_{-i}) > G_i(c_i^{k},a_{-i}).\label{eq:vdominated}
	\end{align}
	By \eqref{eq:sumineq}, such $v_i^k$ exist.
	(Note that the sum in \eqref{eq:vsum0} is taken over $\hat A_{-i}$ and \eqref{eq:vdominated} is required to hold for all action profiles in $A_{-i}$.)
	
	Now define games $G^1,G^2$ on $A$ as follows.
	For all $i\in N$ and $a\in A$,
	\begin{align*}
		G^1_i(a) = 
		\begin{cases}
			G_i(a_i^{k},a_{-i}) + v^{k}(a_{-i})\quad&\text{if } a_i = a_i^{k} \text{ for some $k\in[K]$,}\\
			G_i(b_i^{k},a_{-i}) - v^{k}(a_{-i})\quad&\text{if } a_i = b_i^{k}\text{ for some $k\in[K]$, and}\\
			G_i(c_i^{k},a_{-i})\quad&\text{if $a_i = c_i^{k}$ for some $k\in[K]$.}
		\end{cases}
	\end{align*}
	Define $G^2$ similarly with the roles of $a_i^{k}$ and $b_i^{k}$ exchanged.
	By \eqref{eq:vsum0}, $p$ is an equilibrium of $G^1,G^2$, and by \eqref{eq:vdominated}, all actions in $A_i\setm \hat A_i$ are dominated.
	More specifically, in $G^1$, each $c_i^k$ is dominated by $a_i^k$ and in $G^2$, each $c_i^k$ is dominated by $b_i^k$.
	Hence, the set of rationalizable action profiles in $G^1,G^2$ is $\hat A$.
	By \Cref{lem:conversefullsupportrat}, $p\in f(G^1)\cap f(G^2)$.
	Since $G = \nicefrac12\ G^1 + \nicefrac12\ G^2$, consistency implies that $p\in f(G)$.
\end{proof}

\Cref{lem:quasi-strict} together with consequentialism allows us to push slightly beyond quasi-strict equilibria. 
If $p$ is an equilibrium of a game $G$ so that for every player $i$, every action of $i$ that is a best response against $p_{-i}$ is either in the support of $p_i$ or a clone of such an action, then we get from consequentialism that $p\in f(G)$.
In that case, we say that $p$ is essentially quasi-strict.

\begin{definition}[Essentially quasi-strict equilibrium]\label{def:essstrict}
	Let $G$ be a game on $A$.
	An equilibrium $p$ of $G$ is \emph{essentially quasi-strict} if there is a blow-down $G'$ of $G$ with surjection $\phi$ so that $\phi_*(p)$ is a quasi-strict equilibrium of $G'$.
\end{definition}

Similarly, one could define essentially unique and essentially full support equilibria, but we will not need these notions.
Note that if a solution concept satisfies consequentialism and returns all quasi-strict equilibria, then it also has to return all essentially quasi-strict equilibria.
We use this fact in the proof of the last step: if a solution concept satisfies consequentialism and consistency and returns all quasi-strict equilibria, then it, in fact, has to return all equilibria.

\begin{lemma}[Reduction from quasi-strict to all equilibria]\label{lem:quasistricttoall}
 	Let $f$ be a solution concept that satisfies consequentialism and consistency so that $p\in f(G)$ whenever $p$ is a quasi-strict equilibrium of $G$.
	Then, $\nash\subseteq f$.
\end{lemma}
\begin{proof}
	Let $G$ be a game on $A$ and $p\in\nash(G)$.
	For $i\in N$, let $\hat A_i = \supp(p_i)$ and for any game $G'$ on $A$ with $p\in \nash(G')$, let
	\begin{align*}
		\bar A_i(G') = \{a_i\in A_i\setm\hat A_i\colon G'_i(a_i,p_{-i}) = G'_i(p_i,p_{-i})\text{ and } a_i \text{ is not a clone of an action in } \hat A_i\}.
	\end{align*}
	That is, $\bar A_i(G')$ is the set of actions that are best responses against $p_{-i}$ and not in the support of $p_i$ or clones of actions in the support of $p_i$.
	We write $\bar A_i = \bar A_i(G)$.
	Note that $p$ is an essentially quasi-strict equilibrium of $G'$ if $\bar A_i(G') = \emptyset$ for all $i\in N$.
	
	We prove that $p\in f(G)$ by induction on the number of players for which $\bar A_i\neq\emptyset$.
	If $\bar A_i = \emptyset$ for all $i\in N$, the statement follows from the assumption that $f$ satisfies consequentialism and returns quasi-strict equilibria.
	Otherwise, let $i\in N$ with $\bar A_i\neq\emptyset$.
	We write $G$ as a convex combination of two games $G^1,G^2$ so that $p\in\nash(G^1)\cap \nash(G^2)$, and
	\begin{align}
		\{j\in N\colon \bar A_i(G^l)\neq\emptyset\} \subseteq \{j\in N\colon \bar A_i\neq\emptyset\}\setm\{i\}\label{eq:morequasistrict}
	\end{align}
	for $l = 1,2$.
	
	By \Cref{lem:slicestochasttransform} and the remarks thereafter, we may assume that $p_i$ is the uniform distribution on $\hat A_i$.
	Moreover, by consequentialism, we may assume that $\hat A_i = \{a_i^1,\dots,a_i^K,b_i^1,\dots,b_i^K\}$, $\bar A_i = \{c_i^1,\dots,c_i^K\}$, and $a_i^{k},b_i^k$ are clones in $G$ for all $k\in[K]$.
	Let $G^1,G^2$ be games on $A$ so that for all $j\in N$ and $a\in A$,
	\begin{align*}
		G^1_j(a) &=
		\begin{cases}
			G_j(c_i^{k},a_{-i})\quad&\text{if $a_i = a_i^{k}$ or $a_i = c_i^{k}$ for some $k\in[K]$,}\\
			G_j(a_i^{k},a_{-i}) + G_j(b_i^{k},a_{-i}) - G_j(c_i^{k},a_{-i})&\text{if $a_i = b_i^{k}$ for some $k\in[K]$, and}\\
			G_j(a_i,a_{-i})&\text{if $a_i\in A_i\setm (\hat A_i\cup \bar A_i)$,}
		\end{cases}
	\end{align*}
	and
	\begin{align*}
		G^2_j(a) &=
		\begin{cases}
			G_j(c_i^{k},a_{-i})\quad&\text{if $a_i = b_i^{k}$ or $a_i = c_i^{k}$ for some $k\in[K]$,}\\
			G_j(a_i^{k},a_{-i}) + G_j(b_i^{k},a_{-i}) - G_j(c_i^{k},a_{-i})&\text{if $a_i = a_i^{k}$ for some $k\in[K]$, and}\\
			G_j(a_i,a_{-i})&\text{if $a_i\in A_i\setm (\hat A_i\cup \bar A_i)$.}
		\end{cases} 
	\end{align*}
	Then, $G = \nicefrac12 G^1 + \nicefrac12 G^2$ and for all $k\in[K]$, $a_i^{k},c_i^{k}$ are clones in $G^1$ and $b_i^{k},c_i^{k}$ are clones in $G^2$.
	Moreover, $p\in\nash(G^1)\cap\nash(G^2)$ by the definition of $\bar A_i$.
	To see that $p_i$ is a best response to $p_{-i}$, recall that for all $k\in[K]$, $$G_i(a_i^k,p_{-i}) = G_i(b_i^k,p_{-i}) = G_i(p_i,p_{-i}) = G_i(c_i^k,p_{-i}),$$ since $c_i^k$ is assumed to be in $\bar A_i$.
	So for all $a_i\in A_i$, $G^1_i(a_i,p_{-i}) = G^2_i(a_i,p_{-i}) = G_i(a_i,p_{-i})$.
	Also, for $j\neq i$ and $a_j\in A_j$, we have
	\begin{align*}
		G_j(a_j,p_{-j}) &= \frac1{2K}\sum_{k\in[K]} G_j(a_j,a_i^k,p_{-\{i,j\}}) + G_j(a_j,b_i^k,p_{-\{i,j\}})\\
		&= \frac1{2K}\sum_{k\in[K]} G^1_j(a_j,a_i^k,p_{-\{i,j\}}) + G^1_j(a_j,b_i^k,p_{-\{i,j\}})\\ 
		&= G^1_j(a_j,p_{-j}),
	\end{align*}
	where we use that $p_i$ is uniform on $\hat A_i$ for the first equality, and the definition of $G^1$ for the second equality.
	In particular, $p_j$ is a best response to $p_{-j}$ in $G^1$.
	The same holds for $G^2$ with a similar argument.
	Thus, \eqref{eq:morequasistrict} holds.
	Now, by induction, $p\in f(G^1)\cap f(G^2)$, and so by consistency, we get $p\in f(G)$.
\end{proof}

The fact that any nice total solution concept is a coarsening of \nash now follows from \Cref{lem:quasi-strict} and \Cref{lem:quasistricttoall}.
This completes the proof of \Cref{thm:nash}.

\section{Omitted Proofs From \Cref{sec:robustness}}\label{sec:robustnessproof}

Throughout this section, we consider equivariant solution concepts that are $\delta$-nice for some small enough $\delta$.
Let $G$ be a game with action profiles $A$ and fix a strategy profile $p\in f(G)$. 
\Cref{lem:convex} below then states that a new game $\hat G$ can be obtained by adding an extra action $\hat a_i$ to the action set of every player $i\in N$ such that the following holds.
\begin{enumerate}
	\item There is a profile $q$ close to $p$ so that for all $i\in N$, if $i$ plays $\hat a_i$ in $\hat G$, all players get the same payoff as if $i$ played $q_i$ in $G$. 
	\item There is a profile $\hat p \in f(\hat G)$ so that for all $i\in N$, $\hat p_i$ has all but a small fraction of probability on $\hat a_i$.
\end{enumerate}

\begin{lemma}\label{lem:convex}
	Let $\delta > 0$ and $f$ an equivariant solution concept that satisfies $\delta$-consequentialism and $\delta$-consistency.
	Let $G$ be a game on $A$ and $p\in f(G)$.
	Then, there is a game $\hat G$ with action set $\hat A_i = A_i\cup\{\hat a_i\}$ for all $i\in N$ so that the following holds.
	\begin{enumerate}
		\item There is $q\in B_{3\delta}(p)$ such that $\hat G(\hat a_I, a_{-I}) = G(q_I,a_{-I})$ for all $I \subseteq N$
		 and $a_{-I}\in A_{-I}$.
		\item There is $\hat p\in f(\hat G)$ such that $\hat p_i(\hat a_i) \ge 1 - 3\delta$ for all $i\in N$.
	\end{enumerate}

\end{lemma}
\begin{proof}
	For all $i\in N$, let $k_i = |A_i|\lceil \frac{1}\delta\rceil$.\footnote{For $x\in\mathbb R$, $\lceil x\rceil$ is the smallest integer that is at least as large as $x$. Similarly, $\lfloor x\rfloor$ is the largest integer that is at most as large as $x$.}
	For each $a_i \in A_i$, let $A_i^{a_i}\in\mathcal F(U)$, all disjoint and disjoint from $A_i$, with $|A_i^{a_i}| = k_i$.
	Let $G'$ be a game with action set $A_i' = A_i\cup (\bigcup_{a_i\in A_i} A_i^{a_i})$ for all $i\in N$, and $G$ a blow-down of $G'$ with surjections $\phi_i\colon A_i'\rightarrow A_i$ so that $\phi_i^{-1}(a_i) = \{a_i\}\cup A_i^{a_i}$ for all $a_i\in A_i$ and $i\in N$.
	That is, $G'$ results from $G$ by adding $k_i$ clones of $a_i$ for each action $a_i$.
	 
	Since $f$ satisfies $\delta$-consequentialism, $\phi_*^{-1}(p)\subseteq B_\delta(f(G'))$.
	We construct $p'\in f(G')$ so that $p'_i$ assigns probability at least $1-\delta$ uniformly to some subset $\tilde A_i$ of $\bigcup_{a_i\in A_i} A_i^{a_i}$.
	Let $\tilde p\in B_\delta(p) \cap \phi_*(f(G'))$, which exists since $f$ satisfies $\delta$-consequentialism.
	Let $l_i = \lfloor k_i \tilde p_i\rfloor \in \{0,\dots,k_i\}^{A_i}$ and $r_i = k_i \tilde p_i - l_i\in [0,1)^{A_i}$.
	For each $a_i\in A_i$, choose a subset $\tilde A_i^{a_i}$ of $A_i^{a_i}$ with $|\tilde A_i^{a_i}| = l_i(a_i)$ and let $\tilde A_i = \bigcup_{a_i\in A_i} \tilde A_i^{a_i}$.
	Let $p'$ be a profile on $A'$ so that for all $i\in N$ and $a_i\in A_i$,
	\begin{align*}
		p'_i(a_i') =
		\begin{cases}
			\frac1{k_i} &\text{ for } a_i' \in \tilde A_i^{a_i}\text{, and}\\
			\frac{r_i(a_i)}{k_i}\frac{1}{|\{a_i\}\cup A_i^{a_i} - \tilde A_i^{a_i}|} \qquad&\text{ for } a_i' \in \{a_i\}\cup A_i^{a_i} - \tilde A_i^{a_i}.
		\end{cases}
	\end{align*}
	Observe that for all $a_i\in A_i$,
	\begin{align*}
		p'(\{a_i\}\cup A_i^{a_i}) = \frac{|\tilde A_i^{a_i}|}{k_i} + \frac{r_i(a_i)}{k_i} = \frac{l_i(a_i)}{k_i} + \frac{r_i(a_i)}{k_i} = \tilde p_i(a_i).
	\end{align*}
	Hence, $p'$ is well-defined and $\phi_* p' = \tilde p$.
	By the choice of $\tilde p$, $p'\in f(G')$.

	Recall that $r_i\in [0,1)^{A_i}$, and so
	\begin{align*}
		\frac{\|r_i\|}{k_i} < \frac{|A_i|}{k_i} \le \delta.
	\end{align*}
	Let $q$ be the profile on $A$ with $q_i = \frac{\tilde p_i - r_i}{\|\tilde p_i - r_i\|}$.
	Since $\frac{\|r_i\|}{k_i} < \delta$, it follows that $q\in B_{2\delta}(\tilde p)\subseteq B_{3\delta}(p)$.

	For all $i\in N$, let $\tilde\Sigma_{A_i'}\subseteq \Sigma_{A_i'}$ be the permutations that map the set $\{a_i\}\cup A_i^{a_i} - \tilde A_i$ to itself for all $a_i\in A_i$, and let $\tilde\Sigma_{A'} = \tilde\Sigma_{A_1'}\times\dots\times\tilde\Sigma_{A_n'}$.
	Then, $p' = p'\circ \pi$ for all $\pi\in\tilde\Sigma_{A'}$ since $p'_i$ is a uniform distribution on $\tilde A_i$ and the uniform distribution on $\{a_i\}\cup A_i^{a_i} - \tilde A_i^{a_i}$ for every $a_i\in A_i$ and $i\in N$.
	Thus, equivariance of $f$ implies that $p' = p'\circ\pi\in f(G'\circ \pi)$.
	Let
	\begin{align*}
		\bar G = \frac1{|\tilde\Sigma_{A'}|}\sum_{\pi\in\Pi} G'\circ \pi
	\end{align*}
	It follows from $\delta$-consistency that $p'\in B_\delta(f(\bar G))$.
	Let $\bar p\in f(\bar G)\cap B_\delta(p')$.
	By construction, for all $i\in N$, all actions in $\tilde A_i$ are clones of each other in $\bar G$, and for each $a_i\in A_i$, all actions in $\{a_i\}\cup A_i^{a_i} - \tilde A_i^{a_i}$ are clones of each other.
	Let $\hat G$ be the blow-down of $\bar G$ with action set $\hat A_i = A_i\cup \{\hat a_i\}$ for all $i\in N$ and surjections $\hat\phi_i\colon A_i'\rightarrow \hat A_i$ so that $\hat\phi_i^{-1}(\hat a_i) = \tilde A_i$ and $\hat\phi_i^{-1}(a_i) = \{a_i\}\cup A_i^{a_i}-\tilde A_i^{a_i}$ for all $a_i\in A_i$.
	Since $f$ satisfies $\delta$-consequentialism, there is $\hat p\in f(\hat G)$ so that $\hat\phi_*(\bar p) \in B_\delta(\hat p)$.
	Hence, for all $i\in N$,
	\begin{align*}
		\sum_{a_i\in A_i} \hat p_i(a_i) \le \sum_{a_i\in A_i} \bar p_i(a_i\cup A_i^{a_i} - \tilde A_i^{a_i}) + \delta
		\le \sum_{a_i\in A_i} p'_i(a_i\cup A_i^{a_i} - \tilde A_i^{a_i}) + 2\delta \le 3\delta,
	\end{align*}
	where the last inequality uses the definition of $p_i'$ and $\frac{\|r_i\|}{k_i}\le\delta$.
	Equivalently, $\hat p_i(\hat a_i) \ge 1 - 3\delta$ for all $i\in N$.
	Moreover, by construction of $\hat G$, $\hat G(\hat a_I, a_{-I}) = G(q_I,a_{-I})$ for all $I\subseteq  N$, and $a_{-I}\in A_{-I}$.
\end{proof}

Now consider a game $\hat G$ with action profiles $\hat A$ and let $\hat a\in \hat A$, $a_i\in \hat A_i$, and $i\in N$ such that action $a_i$ yields strictly more payoff against $\hat a_{-i}$ than action $\hat a_i$.
\Cref{lem:dominated} below shows that if there is $\hat p\in f(\hat G)$ so that $\hat p_j$ assigns probability close to $1$ to $\hat a_j$ for each $j\in N$, then there is a game $G$ and $p\in f(G)$ such that $p_i$ assigns probability close to 1 to a dominated action.

\begin{lemma}\label{lem:dominated}
	Let $\eps,\delta> 0$ with $4\left(\lceil\frac1\eps\rceil + 1\right)\delta \le (1-2\delta)$ and $2\delta \le \eps$.
	Let $f$ be an equivariant solution concept that satisfies $\delta$-consequentialism and $\delta$-consistency.
	Let $\hat G$ be a normalized game with action sets $\hat A_j = A_j\cup \{\hat a_j\}$ for all $j\in N$ and let $i\in N$ and $a_i\in A_i$ so that $\hat G_i(a_i,\hat a_{-i}) > \hat G_i(\hat a_i,\hat a_{-i}) + \eps$.
	Then, if $\hat p\in f(\hat G)$ with $\hat p_j(\hat a_j) \ge 1-\delta$ for all $j\in N$, there is a game $\bar G$ and $\bar p\in f(\bar G)$ so that $\bar p_i$ assigns probability at least $1-3\delta$ to an action that is $\delta$-dominated in $\bar G$.
\end{lemma}

\begin{proof}
	Let $M = 2(\lceil\frac1\eps\rceil + 1)$.
	For each $j\in N$, let $A_j = \{a_j^1,\dots,a_j^{|A_j|}\}$.
	Let $G'$ be a game with action set $A_i' = \hat A_i$ for $i$ and $A_j' = \hat A_j\cup \{a_j^{k,l}\colon k\in[|A_j|], l \in[M]\}$ so that each $a_j^{k,l}$ is a clone of $\hat a_j$ for all $j\in N\setm\{i\}$.
	That is, let $\phi_i\colon A_i'\rightarrow\hat A_i$ be the identity, and for all $j\in N\setm \{i\}$, let $\phi_j\colon A_j'\rightarrow \hat A_j$ be the identity on $A_j$ and $\phi_j^{-1}(\hat a_j) = \{\hat a_j\}\cup\{a_j^{k,l}\colon k\in[|A_j|],l\in[M]\}$.
	Then, $\hat G$ is a blow-down of $G'$ with surjection $\phi = (\phi_1,\dots,\phi_n)$.  
	Since $f$ satisfies $\delta$-consequentialism, there is $\hat p'\in B_\delta(\hat p)$ so that $p'\in f(G')$ whenever $\phi_*(p') = \hat p'$.
	In particular, there is $p'\in f(G')$ with 
	\begin{align*}
		\sum_{k = 1}^{|A_j|} p'_j(a_j^k) &\le 2\delta &&\text{ for all $j\in N$, and}\\
		p'_j(a_j^{k,l}) &= \hat p'_j(a_j^k) \qquad&&\text{ for all } k\in[|A_j|], l\in[M]\text{, and } j\in N\setm\{i\}.
	\end{align*}
	The latter condition can be satisfied since 
	\begin{align*}
		M\sum_{k = 1}^{|A_j|}\hat p'(a_j^k) \le 2M\delta \le (1 - 2\delta) \le \hat p'(\hat a_j) = \sum_{a_j\in\phi_j^{-1}(\hat a_j)} p'(a_j),
	\end{align*}
	for all $j\in N\setm\{i\}$.
	Let $\Sigma'_i\subseteq \Sigma_{A_i'}$ be the set consisting of the identity permutation on $U$, and for all $j\in N\setm\{i\}$, let $\Sigma'_j\subseteq \Sigma_{A_j'}$ be the set of permutations that map $\{a_j^k\}\cup \{a_j^{k,l}\colon l\in [M]\}$ to itself for all $k\in[|A_j|]$.
	Hence, a large set of clones of $\hat a_j$ is permuted with each action $a_j^k$ in all possible ways.
	Let $\Sigma' = \Sigma'_1\times\dots\times \Sigma'_n$.
	Note that, by construction, $p'\circ\pi = p'$ for $\pi\in\Sigma'$.
	Let
	\begin{align*}
		\bar G = \frac1{|\Sigma'|}\,\sum_{\pi\in\Sigma'} G'\circ \pi.
	\end{align*}
	For all $j\in N\setm\{i\}$ and $k\in[|A_j|]$, $\{a_j^k\}\cup \{a_j^{k,l}\colon l\in[M]\}$ is a set of clones in $\bar G$.
	Since $\hat G_i(a_i,\hat a_{-i}) > \hat G_i(\hat a_i,\hat a_{-i}) + \eps$, $M = 2(\lceil\frac1\eps\rceil + 1)$, and $\hat G$ is normalized, $\bar G_i(a_i,a_{-i}) > \bar G_i(\hat a_i,a_{-i}) + \frac{\eps}2$ for all $a_{-i}\in A_{-i}'$.
	Hence, $a_i$ $\delta$-dominates $\hat a_i$.
	Lastly, since $f$ satisfies $\delta$-consistency, there is $\bar p \in f(\bar G)$ so that for all $j\in N$,
	\begin{align*}
		\bar p_j(\hat a_j) \ge p_j'(\hat a_j) - \delta \ge 1 - 3\delta.
	\end{align*}
\end{proof}

\begin{proof}[Proof of \Cref{thm:continuityofcharacterization}]
	Given $\eps > 0$, let $\delta > 0$ so that
	\begin{align}
		3(n-1)\delta&\le\frac\eps4,\label{eq:delta1}\\
		4\left(\left\lceil\frac3\eps\right\rceil + 1\right)3\delta&\le (1-6\delta)\text{, and}\label{eq:delta2}\\
		1-9\delta &> \nicefrac12\label{eq:delta3}.
	\end{align}
	Assume that $f$ satisfies $\delta$-consequentialism, $\delta$-consistency, and $\delta$-rationality, but is \emph{not} a refinement of $\nash_\eps$ on the set of normalized games. 
	Then, there is a normalized game $G$ on $A$, $i\in  N$, $a_i\in A_i$, and $p\in f(G)$ so that
	\begin{align}
		G_i(a_i,p_{-i}) > G_i(p_i,p_{-i}) + \eps.\label{eq:notepseq}
	\end{align}
	Observe that for all $q\in B_{3\delta}(p)$ and $p_i'\in\Delta A_i$,
	\begin{align}
		\left|G_i(p_i',p_{-i}) - G_i(p_i',q_{-i})\right| \le 3(n-1)\delta.\label{eq:payoffdiff}
	\end{align}
	By \Cref{lem:convex}, there is a game $\hat G$ with action set $\hat A_j = A_j\cup\{\hat a_j\}$ for all $j\in N$ so that the following holds. 
	\begin{enumerate}
		\item There is $q\in B_{3\delta}(p)$ such that $\hat G(\hat a_I, a_{-I}) = G(q_I,a_{-I})$ for all $I \subseteq N$ and $a_{-I}\in A_{-I}$.\label{item:ghat1}
		\item There is $\hat p\in f(\hat G)$ such that $\hat p_j(\hat a_j) \ge 1 - 3\delta$ for all $j\in N$.\label{item:ghat2}
	\end{enumerate}	
	Applying \ref{item:ghat1} with $I = N\setm\{i\}$, it follows from \eqref{eq:delta1},~\eqref{eq:notepseq}, and \eqref{eq:payoffdiff} that
	\begin{align*}
		\hat G_i(a_i,\hat a_{-i}) &= G_i(a_i,q_{-i}) \ge G_i(a_i,p_{-i}) - 3(n-1)\delta\\ 
		&> G_i(p_i,p_{-i}) + \nicefrac34\eps \ge G_i(p_i,q_{-i}) + \frac24\eps\\
		&\ge G_i(q_i,q_{-i}) + \nicefrac14\eps = \hat G_i(\hat a_i,\hat a_{-i}) + \nicefrac14\eps.
	\end{align*}
	By \eqref{eq:delta1} and \eqref{eq:delta2}, \Cref{lem:dominated} applied to $\frac\eps4$ and $3\delta$ gives a game $G'$ and $\bar p\in f(G')$ so that $\bar p_i$ assigns probability at least $1 - 9\delta$ to a $3\delta$-dominated action.
	Since $1 - 9\delta > \nicefrac12$ by \eqref{eq:delta3}, this contradicts $\delta$-rationality.
\end{proof}

\end{document}